\pdfoutput=1
\RequirePackage{ifpdf}
\ifpdf 
\documentclass[pdftex]{sigma}
\else
\documentclass{sigma}
\fi

\numberwithin{equation}{section}

\newtheorem{Theorem}{Theorem}[section]
\newtheorem*{Theorem*}{Theorem}

 { \theoremstyle{definition}

\newtheorem{Examples}[Theorem]{Examples}

\newtheorem{Conjecture}[Theorem]{Conjecture}
}

\usepackage{eepic}

\usepackage[active]{srcltx}

\newcommand{\pa}{\partial}

\newcommand{\la}{\lambda}
\newcommand{\ta}{\tau}

\newcommand{\om}{\omega}
\newcommand{\de}{\delta}

\newcommand{\ka}{\kappa}
\newcommand{\De}{\Delta}

\newcommand {\R} {{\mathbb R}}
\newcommand {\C} {{\mathbb C}}
\newcommand {\ag} {\mathfrak{a}}

\newcommand{\half}{\frac{1}{2}}
\newcommand{\rar}{\rightarrow}

\begin{document}
\allowdisplaybreaks

\renewcommand{\thefootnote}{}

\newcommand{\arXivNumber}{2305.00529}

\renewcommand{\PaperNumber}{012}

\FirstPageHeading

\ShortArticleName{$\mathfrak{gl}(3)$ Polynomial Integrable System: Different Faces}

\ArticleName{$\boldsymbol{\mathfrak{gl}(3)}$ Polynomial Integrable System: Different\\ Faces of the 3-Body/$\boldsymbol{{\mathcal A}_2}$ Elliptic Calogero Model\footnote{This paper is a~contribution to the Special Issue on Symmetry, Invariants, and their Applications in honor of Peter J.~Olver. The~full collection is available at \href{https://www.emis.de/journals/SIGMA/Olver.html}{https://www.emis.de/journals/SIGMA/Olver.html}}}

\Author{Alexander V.~TURBINER, Juan Carlos LOPEZ VIEYRA \newline and Miguel A.~GUADARRAMA-AYALA}

\AuthorNameForHeading{A.V.~Turbiner, J.C.~Lopez Vieyra and M.A.~Guadarrama-Ayala}

\Address{Instituto de Ciencias Nucleares, Universidad Nacional Aut\'onoma de M\'exico, \\
 Apartado Postal 70-543, 04510 Ciudad de M\'exico, Mexico}
\Email{\href{mailto:turbiner@nucleares.unam.mx}{turbiner@nucleares.unam.mx}, \href{mailto:vieyra@nucleares.unam.mx}{vieyra@nucleares.unam.mx}, \href{mailto:miguel.ayala@mail.mcgill.ca}{miguel.ayala@mail.mcgill.ca}}

\ArticleDates{Received July 26, 2023, in final form January 22, 2024; Published online February 03, 2024}

\Abstract{It is shown that the $\mathfrak{gl}(3)$ polynomial integrable system, introduced by Sokolov--Turbiner in [\textit{J.~Phys.~A} \textbf{48} (2015), 155201, 15~pages, arXiv:1409.7439], is equivalent to the~$\mathfrak{gl}(3)$ quantum Euler--Arnold top in a constant magnetic field. Their Hamiltonian as well as their third-order integral can be rewritten in terms of~$\mathfrak{gl}(3)$ algebra generators. In turn, all these $\mathfrak{gl}(3)$ generators can be represented by the non-linear elements of the universal enveloping algebra of the 5-dimensional Heisenberg algebra $\mathfrak{h}_5(\hat{p}_{1,2},\hat{q}_{1,2}, I)$, thus, the Hamiltonian and integral are two elements of the universal enveloping algebra $U_{\mathfrak{h}_5}$. In this paper, four different representations of the $\mathfrak{h}_5$ Heisenberg algebra are used: (I)~by differential operators in two real (complex) variables, (II)~by finite-difference operators on uniform or exponential lattices. We discovered the existence of two 2-parametric bilinear and trilinear elements (denoted $H$ and $I$, respectively) of the universal enveloping algebra $U(\mathfrak{gl}(3))$ such that their Lie bracket (commutator) can be written as a~linear superposition of {\it nine} so-called {\it artifacts} -- the special bilinear elements of $U(\mathfrak{gl}(3))$, which vanish once the representation of the $\mathfrak{gl}(3)$-algebra generators is written in terms of the $\mathfrak{h}_5(\hat{p}_{1,2},\hat{q}_{1,2}, I)$-algebra generators. In this representation all nine artifacts vanish, two of the above-mentioned elements of $U(\mathfrak{gl}(3))$ (called the Hamiltonian $H$ and the integral $I$) commute(!); in particular, they become the Hamiltonian and the integral of the 3-body elliptic Calogero model, if $(\hat{p},\hat{q})$ are written in the standard coordinate-momentum representation. If $(\hat{p},\hat{q})$ are represented by finite-difference/discrete operators on uniform or exponential lattice, the Hamiltonian and the integral of the 3-body elliptic Calogero model become the isospectral, finite-difference operators on uniform-uniform or exponential-exponential lattices (or mixed) with polynomial coefficients. If $(\hat{p},\hat{q})$ are written in complex $(z, \bar{z})$ variables the Hamiltonian corresponds to a complexification of the 3-body elliptic Calogero model on~${\mathbb C^2}$.}


\Classification{81R12; 81S05; 17J35; 81U15}

\renewcommand{\thefootnote}{\arabic{footnote}}
\setcounter{footnote}{0}

\section{Introduction}

Let us take a finite-dimensional Lie algebra $\mathfrak{g}$ spanned by the generators $J_i$, $i=1,2,\dots, \dim \mathfrak{g}$. The second degree polynomial in the $J$-generators,
\[
 H(J) = \sum_{i,j=1}^{\dim \mathfrak{g}} a_{ij} \{ J_i, J_j\} + \sum_i^{\dim \mathfrak{g}} b_i J_i ,
\]
where $\{ A, B\}=AB + BA$ is the anti-commutator and $\{a\}$, $\{b\}$ are parameters,
defines the Hamiltonian of the quantum Euler--Arnold top in a constant magnetic field with components $b_i$, $i=1,2,\dots, \dim \mathfrak{g}$. It is well known that
the generators $J_i$ of any semi-simple Lie algebra can be written in terms of the generators $(\hat{p},\hat{q})$ of a Heisenberg algebra, hence, $J_i=J_i(\hat{p},\hat{q})$.
We call such a system a $\mathfrak{g}$-polynomial system if its Hamiltonian is defined as
\[
 H(\hat{p},\hat{q}) \equiv H(J(\hat{p},\hat{q})) .
\]

A particular example of an $\mathfrak{sl}(2)$-polynomial system was studied in details in \cite[equation~(13)]{Turbiner:2000},
which is associated with the harmonic oscillator,
\[
 H = - \hat{q}\hat{p}^2 + \bigl(\hat{q} - p - \tfrac{1}{2}\bigr) \hat{p} = -J^0 J^- + J^0 - \bigl(p+\tfrac{1}{2}\bigr)J^- ,
\]
where $p=0,1$ and
\[
 J^0 = \hat{q}\hat{p} ,\qquad J^- = \hat{p} ,
\]
are two $\mathfrak{sl}(2)$
 generators, $\big[J^0, J^-\big]=-J^-$, see below. The general $\mathfrak{sl}(2)$-polynomial system is associated with the Heun operator, which is equivalent to the $BC_1$ elliptic Calogero
model \cite{Heun}.
The present paper is aimed at constructing an analogous but $\mathfrak{gl}(3)$-polynomial system starting from the quantum ${\mathcal A}_2$ elliptic (3-body Calogero) model.

Celebrated 3-body elliptic Calogero model or, stated differently, the ${\mathcal A}_{2}$ elliptic model (in the Hamiltonian reduction nomenclature, see, e.g., \cite{Olshanetsky:1983}), describes three point-like one-dimensional particles of unit masses on the real line with pairwise interaction given by the Weierstrass $\wp$-function with rectangular fundamental domain $(\om, {\rm i} \om_1)$. It is characterized by the Hamiltonian
\begin{gather}
 {\mathcal H}^{(e)}_{{\mathcal A}_2} = - \half \sum_{i=1}^3\frac{\pa^{2}}{\pa x_i^{2}} +
 \nu (\nu -1) (\wp (x_1-x_2) + \wp (x_2-x_3) + \wp (x_3-x_1))
 \nonumber\\ \hphantom{{\mathcal H}^{(e)}_{{\mathcal A}_2}}{}
 \equiv - \half \De^{(3)} + V_{{\mathcal A}_2} ,\label{OPHam}
\end{gather}
where $x_1$, $x_2$, $x_3$ are the coordinates of the bodies, $\De^{(3)}$ is three-dimensional flat Laplace operator, which represents the kinetic energy, $\ka \equiv \nu(\nu-1)$ is the coupling constant. The Weierstrass function $\wp (x) \equiv \wp (x\mid g_2,g_3)$ (see, e.g., \cite{WW:1927}) is defined as the solution of the equation
\begin{equation*}
 \bigl(\wp'(x)\bigr)^2 = 4 \wp^3 (x) - g_2 \wp (x) - g_3 =
 4(\wp (x) -e_1)(\wp (x) -e_2)(\wp (x) -e_3) ,
\end{equation*}
where $g_2$, $g_3$ are the so-called elliptic invariants, which can be conveniently parameterized as follows:
\begin{equation}
\label{wp-inv}
 g_2 = 12 \bigl(\ta^2 - \mu\bigr) ,\qquad g_3=4 \ta \bigl(2\ta^2-3\mu\bigr) ,
\end{equation}
where $\ta$, $\mu$ are parameters, and $e_1$, $e_2$, $e_3$ are its roots which are chosen, conventionally, to obey $e \equiv e_1+e_2+e_3 = 0$. Since the Hamiltonian~(\ref{OPHam}) is translation-invariant, $x \rar x + a$, one can introduce the center-of-mass and relative coordinates,
\begin{equation}
\label{y}
 Y=\sum_1^3 x_i ,\qquad y_i = x_i - \frac{1}{3}Y ,
\end{equation}
with the condition $\sum_1^3 y_i=0$. The Laplacian $\De^{(3)} \equiv
\sum_{i=1}^3\frac{\pa^{2}}{\pa x_i^{2}}$ in these coordinates takes the form,
\[
 \De^{(3)} = 3 \pa_Y^2 +
 \frac{2}{3} \biggl(\frac{\pa^{2}}{\pa y_1^{2}}+\frac{\pa^{2}}{\pa y_2^{2}} -
 \frac{\pa^{2}}{\pa y_1 \pa y_2}\biggr) .
\]
Separating out the center-of-mass coordinate $Y$, the two-dimensional Hamiltonian arises
\begin{gather}
 {\mathcal H}_{{\mathcal A}_2} = - \frac{1}{3} \biggl(\frac{\pa^{2}}{\pa y_1^{2}}+
 \frac{\pa^{2}}{\pa y_2^{2}} - \frac{\pa^{2}}{\pa y_1 \pa y_2}\biggr) \nonumber\\
 \hphantom{{\mathcal H}_{{\mathcal A}_2} =}{} +
 \nu (\nu -1) (\wp (y_1-y_2) + \wp (2y_1+y_2) + \wp (y_1+2y_2)) ,\label{OPHam2}
\end{gather}
which seemingly was already known to Charles Hermite as a two-dimensional generalization
of the celebrated one-dimensional Lam\'e operator
(following Serguei P.~Novikov's studies of unpublished notes by Charles Hermite communicated
to one of the authors (Alexander V.~Turbiner)),
\begin{equation*}
 {\mathcal H}^{(e)}_{{\mathcal A}_1} = - \half \frac{\pa^{2}}{\pa y^{2}} + \ka \wp (y) ,
\end{equation*}
which is also the Hamiltonian of the ${\mathcal A}_1$ elliptic model \cite{Olshanetsky:1983},
see also \cite{Lame:1989}.
We will call the operator~(\ref{OPHam2}) {\it the two-dimensional Lam\'e operator}. In general, the above procedure allows us to connect the quantum dynamics in the relative space of the three-body problem with two-dimensional quantum dynamics \cite{TME:2016}.

For many years the question of the existence of polynomial eigenfunctions
of the operator~(\ref{OPHam2}) was a challenge to answer. It was eventually solved
in 2015 by Sokolov--Turbiner in \cite{ST:2015}: the discrete values
of the coupling constant were found
\begin{equation}
\label{kappa}
 \ka \equiv \nu(\nu-1) = \frac{n}{9} (n+3) ,\qquad n=0,1,2,\dots,
\end{equation}
for which the $\frac{(n+2)(n+1)}{2}$ polynomial eigenfunctions exist in the variables
\begin{equation}
\label{trans}
x=\frac{f'(y_1)-f'(y_2)}{f(y_1) f'(y_2)-f(y_2) f'(y_1)} , \qquad
y=\frac{2(f(y_1)- f(y_2))}{f(y_1) f'(y_2)-f(y_2) f'(y_1)} ,
\end{equation}
where
\[
 f(x) = \wp (x | g_2, g_3) + \ta ,
\]
is the shifted Weierstrass function.

In very tedious and highly non-trivial calculations, performed in \cite{ST:2015},
it was found that the~${\mathcal A}_2$ elliptic Calogero model potential $V_{{\mathcal A}_2}$ (see~(\ref{OPHam}) and~(\ref{OPHam2})) in variables~(\ref{trans}) takes the form of ratio of polynomials,
\begin{equation*}
 V_{{\mathcal A}_2}(x,y) = {3\nu(\nu-1)} \frac{\bigl(x+2\tau x^2+\mu x^3-6 \bigl(\mu-\tau^2\bigr) y^2 + 3 \mu \ta x y^2\bigr)^2}{4D} ,
\end{equation*}
where the denominator
\begin{gather}
 4 D(x,y) = 3 \mu^2 x^4 y^2 + 18 \ta \mu^2 x^2 y^4 + 9\mu^2(3\ta^2 - 4 \mu) y^6 -
 4 \mu x^5 - 24 \ta \mu x^3 y^2 \nonumber \\
 \hphantom{4 D(x,y) =}{}
 - 36\mu (\ta^2 - 2\mu) x y^4 -
 4 \ta x^4 - 6 (4\ta^2 + 5\mu) x^2 y^2
 - 18 \ta(2\ta^2 - 3 \mu) y^4 \nonumber\\
 \hphantom{4 D(x,y) =}{}
 - 36 \ta x y^2 - \frac{4}{3} x^3 - 27 y^2 \label{D}
\end{gather}
was called the {\it determinant} or the {\it elliptic discriminant}. In rational limit $\ta=\mu=0$ this is the square of Vandermonde determinant or the discriminant of the cubic equation, in trigonometric limit $\mu=0$ it becomes the so-called trigonometric Vandermonde determinant or a trigonometric discriminant, see for details \cite{ST:2015}. Furthermore, the two-dimensional flat Laplacian in~(\ref{OPHam2}) becomes the Laplace--Beltrami operator in $(x,y)$-coordinates
\begin{gather*}
\De_g(x,y;\ta,\mu) = 3\biggl(\frac{x}{3}+\tau x^2+\mu x^3 +\bigl(\mu-\ta^2\bigr) y^2- \mu \ta x y^2 -
 \mu^2 x^2 y^2 \biggr)\frac{\pa^2}{\pa x^2}\nonumber \\
\hphantom{\De_g(x,y;\ta,\mu) =}{}
 + y \bigl(3 +8 \ta x+7 \mu x^2 - 3\mu \ta y^2 - 6 \mu^2 x y^2 \bigr)\frac{\pa^2}{\pa x \pa y} \nonumber\\
\hphantom{\De_g(x,y;\ta,\mu) =}{}
 + \biggl(-\frac{x^2}{3}+3\ta y^2+4\mu x y^2 -3\mu^2 y^4 \biggr)\frac{\pa^2}{\pa y^2} \nonumber \\
\hphantom{\De_g(x,y;\ta,\mu) =}{}
 + \bigl(1+4 \ta x+5 \mu x^2-3\mu \ta y^2-6\mu^2 x y^2 \bigr) \frac{\pa}{\pa x} + 2 y \bigl(2 \tau+ 3 \mu x-3\mu^2 y^2 \bigr)\frac{\pa}{\pa y} ,
\end{gather*}
with naturally-defined flat contravariant metric $g^{ij}$, $i,j=1,2$ with polynomial entries. It can be easily checked that, remarkably, expression~(\ref{D}) is equal to the determinant of this contravariant metric,
\[
 D={\rm Det} \bigl(g^{ij}\bigr) ,
\]
which explains the name {\it determinant}, used in \cite{ST:2015}.

Surprisingly, the gauge rotation of the 2-dimensional Lam\'e operator~(\ref{OPHam2}) with the determinant $D$~(\ref{D}) to the power $\nu/2$ as a gauge factor transforms operator~(\ref{OPHam2}) into the algebraic operator(!) with polynomial coefficients,
\begin{gather}
 h_{{\mathcal A}_2}(x,y) = -
 3 D^{-\frac{\nu}{2}} ({\mathcal H}_{\rm {\mathcal A}_2} - 3\nu(3\nu+1)\tau) D^{\frac{\nu}{2}}
\nonumber \\ \hphantom{h_{{\mathcal A}_2}(x,y)}{}
= \bigl(x+3 \tau x^2+3 \mu x^3 +3 \bigl(\mu-\ta^2\bigr) y^2-3 \mu \ta x y^2
 -3\mu^2 x^2 y^2 \bigr)\frac{\pa^2}{\pa x^2}
\nonumber \\ \hphantom{h_{{\mathcal A}_2}(x,y) =}{}
+ y \bigl(3+8\ta x+7\mu x^2-3\mu \ta y^2-6\mu^2 x y^2 \bigr)\frac{\pa^2}{\pa x \pa y}
\nonumber \\ \hphantom{h_{{\mathcal A}_2}(x,y) =}{}
+ \frac{1}{3}\bigl(-x^2+9 \tau y^2+12 \mu x y^2-9\mu^2 y^4 \bigr)\frac{\pa^2}{\pa y^2}
\nonumber \\ \hphantom{h_{{\mathcal A}_2}(x,y) =}{}
+ (1+3\nu) \bigl(1+4 \ta x+5 \mu x^2-3\mu \ta y^2-6\mu^2 x y^2 \bigr) \frac{\pa}{\pa x}
\nonumber \\ \hphantom{h_{{\mathcal A}_2}(x,y) =}{}
+ 2 (1+3\nu) y\bigl(2 \tau+ 3 \mu x-3\mu^2 y^2 \bigr)\frac{\pa}{\pa y}
 + 3\nu(1+3\nu) \mu \bigl(2 x-3\mu y^2 \bigr) .\label{hA2diff}
\end{gather}
This was one of the principal results obtained in the article \cite{ST:2015}, which will be essential in the present article. Let us emphasize that the operator $h_{{\mathcal A}_2}(x,y)$ looks like the two-dimensional generalization of the (algebraic) Heun operator, see, e.g., \cite{Heun}.

It was also found in \cite{ST:2015} that the second-order algebraic differential operator $h_{{\mathcal A}_2}(x,y)$ commutes with a non-trivial third-order algebraic differential operator $k_{\rm {\mathcal A}_2}$ with polynomial coefficients,
\[
 [h_{{\mathcal A}_2}(x,y), k_{\rm {\mathcal A}_2}(x,y)] = 0 ,
\]
where
\begin{gather}
k_{\rm {\mathcal A}_2}(x,y) = 2\nu(1+3\nu)(2+3\nu) \mu y \bigl(2\ta + 3\mu x - 3\mu^2 y^2\bigr)
\nonumber \\ \hphantom{k_{\rm {\mathcal A}_2}(x,y) =}{}
 +\frac{1}{3} (1+3\nu)(2+3\nu) y
\bigl(\mu + 8 \ta^2 +28\mu \ta x + 21\mu^2 x^2 - 9\mu^2\ta y^2 - 18\mu^3 x y^2\bigr)
\frac{\pa}{\pa x}
\nonumber \\ \hphantom{k_{\rm {\mathcal A}_2}(x,y) =}{}
-\frac{2}{9} (1+3\nu)(2+3\nu)
\bigl(1 +4\tau x + 6\mu x^2 - 24\mu \ta y^2 - 36\mu^{2} x {y}^{2} + 27\mu^3 y^4\bigr)
\frac{\pa}{\pa y}
\nonumber \\ \hphantom{k_{\rm {\mathcal A}_2}(x,y) =}{}
 + (2+3\nu) y \bigl(3 \ta + 4\bigl(2\ta^2+\mu\bigr)x + 17\mu \ta x^2 + 8\mu^2 x^3
\nonumber \\ \hphantom{k_{\rm {\mathcal A}_2}(x,y) =}{}
 \quad{}+ 3\mu({\tau}^{2} -2\mu) y^2 - 6\mu^{2}\ta x y^2 - 6{\mu}^3 x^2 y^2\bigr)
\frac{\pa^2}{\pa x^2}
\nonumber \\ \hphantom{k_{\rm {\mathcal A}_2}(x,y) =}{}
-\frac{2}{3} (2+3\nu)
\bigl(x+4\ta x^2 + 5\mu x^3 + 3\bigl(\mu - 4\ta^2\bigr) y^2 -27\mu^2 x^2 y^2
\nonumber \\ \hphantom{k_{\rm {\mathcal A}_2}(x,y) =}{}
\quad{} - 33\mu \ta x y^2
+ 9\mu^2 \ta y^4 + 18\mu^3 x y^4 \bigr)
\frac{\pa^2}{\pa x \pa y}
\nonumber \\ \hphantom{k_{\rm {\mathcal A}_2}(x,y) =}{}
 - (2+3\nu)y\biggl(1+\frac{8}{3}\tau x+3\mu x^2 -7\mu \ta y^2 - 10\mu^2 x y^2 + 6\mu^3 y^4\biggr)
\frac{\pa^2}{\pa y^2}
\nonumber \\ \hphantom{k_{\rm {\mathcal A}_2}(x,y) =}{}
 +y \bigl(1 + 5\ta x + 2\bigl(2\mu +3\ta^2\bigr) x^2 + 3\mu (\ta^2 - 2\mu) x y^2 + 9 \mu \ta x^3
\nonumber \\ \hphantom{k_{\rm {\mathcal A}_2}(x,y) =}{}
\quad{} - \ta\bigl(3\mu - 2\ta^2\bigr) y^2 + 3\mu^{2} x^4 - 3\mu^{2}\ta x^2 y^2 - 2\mu^3 x^3 y^2 \bigr)
\frac{\pa^3}{\pa x^3}
\nonumber \\ \hphantom{k_{\rm {\mathcal A}_2}(x,y) =}{}
+\biggl(
 -\frac{2}{3} x^2+ 2\bigl(5{\tau}^{2}+\mu\bigr) x y^2 - 2\ta x^3 + 3\ta y^2-2\mu x^4 +
 3\mu\bigl(\ta^{2}-2\mu\bigr) y^4
\nonumber \\ \hphantom{k_{\rm {\mathcal A}_2}(x,y) =}{}
\quad{} + 19\mu \tau x^2 y^2-6{\mu}^3 x^2 y^4 + 10\mu^{2} x^3 y^2 -6 \mu^2 \ta x y^4 \biggr)
\frac{\pa^3}{\pa x^2 \pa y}
\nonumber \\ \hphantom{k_{\rm {\mathcal A}_2}(x,y) =}{}
- y \biggl( x + \frac{10}{3} \tau {x}^{2} +
 \frac{11}{3} \mu {x}^{3} - 13\mu \ta x y^2 + 3\bigl(\mu-2\ta^2\bigr) y^2
 -11\mu^2 x^2 y^2
\nonumber \\ \hphantom{k_{\rm {\mathcal A}_2}(x,y) =}{}
\quad{} + 3\mu^2\ta y^4 + 6\mu^3 x y^4 \biggr)
\frac{\pa^3}{\pa x \pa y^2}
\nonumber \\ \hphantom{k_{\rm {\mathcal A}_2}(x,y) =}{}
- \biggl(y^2 + \frac{2}{27} x^3 + 2\tau x y^2 - 3\mu \ta y^4 +
 \frac{5}{3}\mu x^2 y^2 - 4{\mu}^{2} x y^4 + 2{\mu}^{3}y^6 \biggr)
\frac{\pa^3}{\pa y^3} .
\label{kA2diff}
\end{gather}
Hence, $h_{{\mathcal A}_2}(x,y)$ and $k_{{\mathcal A}_2}(x,y)$ span the two-dimensional commutative algebra
of the differential operators in two variables, which depend on three free parameters
$\nu$, $\mu$, $\tau$. This is the first non-trivial example of this. Naturally, the third-order differential operator $k_{{\mathcal A}_2}(x,y)$ can be called the {\it integral}. By making the inverse gauge rotation of the integral $k_{{\mathcal A}_2}(x,y)$,
\[
 D^{\frac{\nu}{2}}k_{{\mathcal A}_2}(x,y)D^{-\frac{\nu}{2}} ,
\]
with the determinant $D$~(\ref{D}) to the power $(-\nu/2)$ as the gauge factor and changing variables $(x, y) \rar (y_1, y_2)$~(\ref{trans}), we should arrive at the third-order integral
of the quantum 3-body elliptic Calogero model in the form of the third-order differential
operator with elliptic coefficients found by Oshima \cite{Oshima:2007}. This demonstrates
explicitly the integrability of the original 3-body elliptic Calogero model written
in $y_1$, $y_2$ coordinates.

It was concluded in \cite{ST:2015} that the 3-body elliptic Calogero model defines
a polynomial integrable model with the algebraic Hamiltonian~(\ref{hA2diff})
and the algebraic integral~(\ref{kA2diff}) with $\mu, \ta, \nu$-dependent parametric coefficients.
This model has $\mathfrak{sl}(3)$ hidden algebra in the representation~$(-3\nu, 0)$.
As a result the $\mathfrak{sl}(3)$ quantum Euler--Arnold top in a constant magnetic field occurs.
Note that for discrete values of the coupling constant $\ka$: $n=-3\nu$, $n=0,1,2,\dots$,
the $\mathfrak{sl}(3)$ hidden algebra emerges in the finite-dimensional representation, thus,
the top has a~common finite-dimensional invariant subspace for both $h$ and $k$.

The goal of this article is two-fold. First of all, the above-mentioned polynomial
integrable model, realized in terms of differential operators, will be rewritten
in terms of the generators of the Heisenberg algebra $\mathfrak{h}_5$. Hence, its Hamiltonian will appear as an element of the universal enveloping algebra $U_{\mathfrak{h}_5}$.
Then we project it into the translation-invariant or dilatation-invariant operators
defining two families of 3-parametric $\mu$, $\tau$, $\nu$ isospectral polynomial integrable models on two-dimensional uniform or exponential lattices, respectively, and
two additional families on mixed two-dimensional translation-invariant and dilatation-invariant lattices. All four families admit
2-parametric $\mu$, $\tau$ polynomial eigenfunctions for certain discrete values
of the coupling constant $\nu$. An extra polynomial integrable model occurs as a result of
a special complexification of $\R^2$ to $\C^2$ via the Heisenberg algebra $\mathfrak{h}_5$ generators
acting on the two-dimensional Hilbert space with the Gaussian measure. The spectrum of this
model is characterized by infinite multiplicity and for certain discrete values
of the coupling constant $\ka$~(\ref{kappa}) the eigenfunctions are poly-analytic functions in two complex variables of the fixed degree.
Second of all, it will be shown that $\mathfrak{gl}(3)$ polynomial integrable model, defined
in the Fock space, is related with special bilinear and trilinear, 2-parametric
elements of the universal enveloping algebra of the algebra $\mathfrak{gl}(3)$. It turns out that
these non-linear elements commute once they are written in terms of {\it any} concrete realization of the algebra $\mathfrak{gl}(3)$ by elements of the universal enveloping algebra
$U_{\mathfrak{h}_5}$.

The article is organized with introduction, Sections~\ref{sec:2}--\ref{sec:7}, conclusions and three Appendices~\ref{app:a},~\ref{app:b},~\ref{app:c}.
In Section~\ref{sec:2}, the 3-body elliptic Calogero model in algebraic form is reformulated in Fock space and its $\mathfrak{gl}(3)$-polynomial integrable model is defined. Section~\ref{sec:3} contains four lattice versions of the 3-body elliptic Calogero model. Section~\ref{sec:4} is dedicated to complexification of the $\mathfrak{gl}(3)$-polynomial integrable model into $\C^2$. In Section~\ref{sec:5}, all nine artifacts
of the $\mathfrak{gl}(3)$ algebra are presented as bilinear combinations of the $\mathfrak{gl}(3)$ generators and Theorem is proved that all of them vanish if the $\mathfrak{gl}(3)$ generators are written as non-linear elements of the universal enveloping algebra $U_{\mathfrak{h}_5}$. Section~\ref{sec:6} contains the explicit expressions of the Hamiltonian, the cubic integral and their commutator in terms of the $\mathfrak{gl}(3)$-algebra generators.
In Section~\ref{sec:7}, the $G_2$/3-body (with pairwise and 3-body interactions) elliptic problem is briefly discussed and the Fock space representation of the $G_2$ elliptic 3-body Hamiltonian is constructed.

Throughout the remaining text the {\it hats} in $p,q$'s will be dropped: $({\hat p}, {\hat q}) \rar (p,q)$.

\section{3-body elliptic Calogero model in the Fock space}\label{sec:2}

Let us take 5-dimensional Heisenberg algebra $\mathfrak{h}_5$ spanned by the generators~$p_x$,~$p_y$,~$q_x$,~$q_x$,~$I$, which obey the commutation relations,
\begin{alignat}{5}
 &[p_x, q_x]=1 ,\qquad&& [p_y, q_y]=1 ,\qquad&& [p_x, q_y]=0 ,\qquad&& [p_y, q_x]=0 ,&\nonumber \\
 &[p_x, p_y]=0 ,\qquad&& [q_x, q_y]=0 ,\qquad&& [p_{x,y}, I]=0 ,\qquad&& [q_{x,y}, I]=0 ,& \label{h5}
\end{alignat}
see Appendix~\ref{app:a3}. The universal enveloping algebra of the algebra $\mathfrak{h}_5$: $U_{\mathfrak{h}_5}$, is spanned
by all ordered monomials in $p_x$, $p_y$, $q_x$, $q_y$.

Now let us form in $U_{\mathfrak{h}_5}$ a second degree polynomial in $p$-generators,
\begin{align}
 h_{{\mathcal A}_2}(p_x, q_x, p_y, q_y)&{}=
\bigl(q_x+3 \tau q_x^2+3 \mu q_x^3 +3 \bigl(\mu-\ta^2\bigr) q_y^2-3 \mu \ta q_x q_y^2
 -3\mu^2 q_x^2 q_y^2 \bigr)p_x^2
\nonumber \\
 &\quad{}+ q_y \bigl(3+8\ta q_x+7\mu q_x^2-3\mu \ta q_y^2-6\mu^2 q_x q_y^2 \bigr) p_x p_y
\nonumber \\
 &\quad{}+ \frac{1}{3}\bigl(-q_x^2+9 \tau q_y^2+12 \mu q_x q_y^2-9\mu^2 q_y^4 \bigr)p_y^2
\nonumber\\
 &\quad{}+ (1+3\nu) \bigl(1+4 \ta q_x+5 \mu q_x^2-3\mu \ta q_y^2-6\mu^2 q_x q_y^2 \bigr) p_x
\nonumber\\
 &\quad{} + 2 (1+3\nu) q_y\bigl(2 \tau+ 3 \mu q_x-3\mu^2 q_y^2 \bigr)p_y+ 3\nu(1+3\nu) \mu \bigl(2 q_x-3\mu q_y^2 \bigr)
\nonumber \\
 &{} \equiv \sum_{i,j=x,y} c_{ij}(q) p_i p_j + 
 \sum_{i=x,y} c_{i}(q) p_i + c_0(q) ,
\label{hA2-pq}
\end{align}
where $\ta$, $\mu$, $\nu$ are parameters. Here the coefficients $c_{ij}$ are the 4th degree polynomials in $q$-generators, $c_{i}$ are the 3rd degree ones and $c_0$ is the 2nd degree polynomial. We also form another non-linear combination in $p,q$-generators
in the $U_{\mathfrak{h}_5}$,
\begin{gather}
k_{{\mathcal A}_2}(p_x, q_x, p_y, q_y) = 2\nu(1+3\nu)(2+3\nu) \mu q_y \bigl(2\ta + 3\mu q_x - 3\mu^2 q_y^2\bigr)
\nonumber \\
\qquad{} +\frac{1}{3} (1+3\nu)(2+3\nu) q_y
\bigl(\mu + 8 \ta^2 +28\mu \ta q_x + 21\mu^2 q_x^2 - 9\mu^2\ta q_y^2 - 18\mu^3 q_x q_y^2\bigr)
p_x
\nonumber \\ 
\qquad{}-\frac{2}{9} (1+3\nu)(2+3\nu)
\bigl(1 +4\tau q_x + 6\mu q_x^2 - 24\mu \ta q_y^2 - 36\mu^{2} q_x {q_y}^{2} + 27\mu^3 q_y^4\bigr)
p_y
\nonumber \\ 
\qquad{} + (2+3\nu) q_y \bigl(3 \ta + 4\bigl(2\ta^2+\mu\bigr)q_x + 17\mu \ta q_x^2 + 8\mu^2 q_x^3
\nonumber \\
\qquad\quad{} + 3\mu\bigl({\tau}^{2} -2\mu\bigr) q_y^2 - 6\mu^{2}\ta q_x q_y^2 - 6{\mu}^3 q_x^2 q_y^2\bigr)
p_x^2
\nonumber \\ 
\qquad{}-\frac{2}{3} (2+3\nu)
\bigl(q_x+4\ta q_x^2 + 5\mu q_x^3 + 3\bigl(\mu - 4\ta^2\bigr) q_y^2 -27\mu^2 q_x^2 q_y^2
\nonumber \\
\qquad\quad{} - 33\mu \ta q_x q_y^2
+ 9\mu^2 \ta q_y^4 + 18\mu^3 q_x q_y^4 \bigr)
p_x p_y
\nonumber \\ 
\qquad{} - (2+3\nu)q_y\biggl(1+\frac{8}{3}\tau q_x+3\mu q_x^2 -7\mu \ta q_y^2 - 10\mu^2 q_x q_y^2 + 6\mu^3 q_y^4\biggr)
p_y^2
\nonumber \\ 
\qquad{} +q_y \bigl(1 + 5\ta q_x + 2\bigl(2\mu +3\ta^2\bigr) q_x^2 + 3\mu \bigl(\ta^2 - 2\mu\bigr) q_x q_y^2 + 9 \mu \ta q_x^3
\nonumber \\
\qquad\quad{} - \ta\bigl(3\mu - 2\ta^2\bigr) q_y^2 + 3\mu^{2} q_x^4 - 3\mu^{2}\ta q_x^2 q_y^2 - 2\mu^3 q_x^3 q_y^2 \bigr)
p_x^3
\nonumber \\ 
\qquad{}+\biggl(
 -\frac{2}{3} q_x^2+ 2\bigl(5{\tau}^{2}+\mu\bigr) q_x q_y^2 - 2\ta q_x^3 + 3\ta q_y^2-2\mu q_x^4 +
 3\mu(\ta^{2}-2\mu) q_y^4
\nonumber \\
\qquad\quad{} + 19\mu \tau q_x^2 q_y^2 -6{\mu}^3 q_x^2 q_y^4 + 10\mu^{2} q_x^3 q_y^2 -6 \mu^2 \ta q_x q_y^4 \biggr)
p_x^2 p_y
\nonumber \\ 
\qquad{} - q_y \biggl( q_x + \frac{10}{3} \tau {q_x}^{2} +
 \frac{11}{3} \mu {q_x}^{3} - 13\mu \ta q_x q_y^2 + 3(\mu-2\ta^2) q_y^2
 -11\mu^2 q_x^2 q_y^2
\nonumber \\
\qquad\quad{} + 3\mu^2\ta q_y^4 + 6\mu^3 q_x q_y^4 \biggr)
p_x p_y^2
\nonumber \\ 
\qquad{} - \biggl(q_y^2 + \frac{2}{27} q_x^3 + 2\tau q_x q_y^2 - 3\mu \ta q_y^4 +
 \frac{5}{3}\mu q_x^2 q_y^2 - 4{\mu}^{2} q_x q_y^4 + 2{\mu}^{3}q_y^6 \biggr)
p_y^3
\nonumber \\
\qquad{}\equiv\sum_{i,j,k=x,y} d_{ijk}(q) p_i p_j p_k +
\sum_{i,j=x,y} d_{ij}(q) p_i p_j +
 \sum_{i=x,y} d_{i}(q) p_i + d_0 ,
\label{kA2-pq}
\end{gather}
where the coefficients $d_{ijk}$, $d_{ij}$, $d_{i}$, $d_{0}$ are polynomials in $q$ of degrees 6, 5, 4 and 3, respectively.

\begin{Theorem}\label{t1}
 The expressions~\eqref{hA2-pq} and~\eqref{kA2-pq} form the commutative pair,
\[
 [h_{{\mathcal A}_2}, k_{{\mathcal A}_2}] = 0 ,
\]
for any values of parameters $\ta$, $\mu$, $\nu$.
\end{Theorem}
\begin{proof} By direct calculation.
\end{proof}

Note, that it can be checked by direct calculation that $h_{{\mathcal A}_2}$, $k_{{\mathcal A}_2}$ written in the (classical) phase space $(p,q)$-variables, $\{p, q\} = 1$, do {\it not} form the commutative pair with respect to the Poisson bracket, $\{h_{{\mathcal A}_2}, k_{{\mathcal A}_2}\} \neq 0$, for any values of the parameters $\ta$, $\mu$, $\nu$. The reason for it is conceptual. Since we do not discuss in this paper the classical-quantum correspondence being fully focused on the quantum case, this will be discussed elsewhere.

From one side, it can be easily checked by direct calculation when $(p,q)$ generators of $\mathfrak{h}_5$ are written in the coordinate-momentum representation~(\ref{h3-pq}) the expressions~(\ref{hA2-pq}) and~(\ref{kA2-pq}) become~(\ref{hA2diff}) and~(\ref{kA2diff}), respectively. From another side, the expressions~(\ref{hA2-pq}) and~(\ref{kA2-pq}) can be written in terms of $\mathfrak{gl}(3)$ generators in $(-3\nu, 0)$ representation~(\ref{gl3pq}) as bilinear and trilinear combinations with $\nu$-dependent coefficients, respectively, cf. \cite[equations~(20) and (25)]{ST:2015}. Hence, the expressions~(\ref{hA2-pq}) and~(\ref{kA2-pq}) define the integrable $\mathfrak{gl}(3)$ Euler--Arnold quantum top, or, equivalently, the integrable $\mathfrak{sl}(3)$ Euler--Arnold quantum top of spin $(-3\nu)$.

By introducing the vacuum $|0\rangle$ as an object annihilated by $p$-operators:
\[
 p_x\,|0\rangle = 0 ,\qquad p_y\,|0\rangle = 0 ,
\]
in addition to the universal enveloping algebra $U_{\mathfrak{h}_5}$, this leads to definition of the Fock space. The formal eigenvalue problem in the Fock space for the Hamiltonian $h_{{\mathcal A}_2}$ is as follows:
\begin{equation}
\label{h-Fock}
 h_{{\mathcal A}_2}(p_x, q_x, p_y, q_y) \phi^{(h)}(q_x, q_y) |0\rangle = \la^{(h)} \phi^{(h)}(q_x, q_y) |0\rangle ,
\end{equation}
where $\phi^{(h)}(q)$ is the eigen-operator and $\la^{(h)}$ is the eigenvalue (spectral parameter).
Analogously,%
\begin{equation}
\label{k-Fock}
 k_{{\mathcal A}_2}(p_x, q_x, p_y, q_y) \phi^{(k)}(q_x, q_y) |0\rangle = \la^{(k)} \phi^{(k)}(q_x, q_y) |0\rangle .
\end{equation}
Owing to Theorem~\ref{t1}
the eigenvalue problems~(\ref{h-Fock}) and~(\ref{k-Fock}) have common eigen-operators~$\phi$. If spin $-\nu=n/3$, $n=0,1,2,\dots$, which corresponds to the $\mathfrak{gl}(3)$ finite-dimensional representation $(n,0)$,
the eigenvalue problems~(\ref{h-Fock}) and~(\ref{k-Fock}) have $\frac{(n+2)(n+1)}{2}$ polynomial
eigen-operators~$\phi^{(h,k)}(q_x, q_y)$.

\begin{Examples}\label{examples}\quad
\begin{itemize}\itemsep=0pt
\item For $n=0$ (thus, at zero coupling, $\ka=0$),
\[
 \la^{(h)}_{0,1} = 0 ,\qquad \phi^{(h)}_{0,1} = 1 .
\]

\item For $n=1$ at coupling
\[
 \ka = \frac{4}{9} ,
\]
the operator $h_{A_2}$ has a three-dimensional kernel (three zero modes) of the type
\[
(a_1 q_x+a_2 q_y+b)
\]
with coefficients $a_1$, $a_2$ which do not vanish simultaneously and,
\[
 \la^{(h)}_{1,i} = 0 ,\qquad i=1,2,3 .
\]

\item The first non-zero eigenvalues appear for $n=2$, thus, at
\[
 \ka = \frac{10}{9} .
\]
In total, there exist six polynomial eigenstates. Eigenvalues are given
by the roots of the factorized algebraic equation of degree 6,
\begin{gather*}
 \bigl(\bigl(\la^{(h)}\bigr)^2+ 4 \ta \la^{(h)} + 4 \mu\bigr)
 \bigl(\bigl(\la^{(h)}\bigr)^2+ 8 \ta \la^{(h)} + 4 \mu + 12 \ta^2\bigr) \\
 \qquad{}\times
 \bigl(\bigl(\la^{(h)}\bigr)^2+ 12 \ta \la^{(h)} + 4 \mu + 16 \ta^2\bigr)
 = 0 .
\end{gather*}
Explicitly,
\begin{gather*}
 \bigl({\la^{(h)}}\bigr)_{\pm}^{(1)} = - 2\bigl(\ta \pm \sqrt {\ta^2 - \mu}\bigr) ,\qquad
 \bigl({\la^{(h)}}\bigr)_{\pm}^{(2)} = - 2\bigl(2\ta \pm \sqrt {\ta^2 - \mu}\bigr) ,\\
 \bigl({\la^{(h)}}\bigr)_{\pm}^{(3)} = - 2\bigl(3\ta \pm \sqrt {5\ta^2 - \mu}\bigr) ,
\end{gather*}
and the corresponding six eigen-operators are of the form
\[
\bigl(a_1 q_x^2+a_2 q_x q_y + a_3 q_y^2 b_1 q_x+b_2 q_y+c\bigr)
\]
with parameters $a_1$, $a_2$, $a_3$, which do not vanish simultaneously, and $b_1$, $b_2$, $c$. In the limit $\ta=\mu=0$ (it corresponds to the rational ${\mathcal A}_2$ integrable model without the harmonic oscillator terms) all six eigenvalues are degenerate to zero.
\end{itemize}
\end{Examples}

\section[gl(3)-polynomial integrable model on a lattice]{$\boldsymbol{\mathfrak{gl}(3)}$-polynomial integrable model on a lattice}\label{sec:3}

\subsection{Uniform translation-invariant lattice}

Let us introduce the shift operator,
\[
 T_{\de} f(x) = f(x + \de) ,\qquad T_{\de}={\rm e}^{\de \pa_x} ,
\]
where $\de \in {\R}$ is parameter, which, sometimes, is called {\it spacing}, and
construct a canonical pair of shift operators (see, e.g., \cite{Turbiner:1995})
\begin{equation}
\label{delta}
 D_{\de} = \frac{T_{\de} - 1}{\de} , \qquad X_{\de} = x T_{-\de} =
 x(1-\de{ D}_{-\de}) ,
\end{equation}
where the operator $D_{\de}$ is defined as
\[
 D_{\de} f(x) = \frac{f(x + \de) - f(x)}{\de} ,
\]
sometimes, it is called the Norlund derivative. It is easy to check that $[D_{\de}, X_{\de}]=1$, hence, $D_{\de}$, $X_{\de}$ form the canonical pair, both operators
are non-local. In the limit $\de \rar 0$ this pair becomes the well-known
coordinate-momentum representation $(\pa_x, x)$ of the Heisenberg algebra $\mathfrak{h}_3 (p,q,I)$,
\[
 [p, q]=1 ,\qquad [p, I]=[q, I]=0 .
\]

For non-vanishing $\de$, the canonical pair~(\ref{delta}) belongs to the extended universal enveloping algebra ${\hat U}_{\mathfrak{h}_3}$. These operators act on infinite uniform lattice space
with spacing $\de$
\[
 \{ \dots, x-2\de , x -\de ,x , x+\de , x+2\de , \dots \}
\]
marked by $x \in {\R}$ -- a position of a central (or reference) point of the lattice.

By taking $D_{\de}$, $X_{\de}$~(\ref{delta}) as basic elements, it can be shown that algebra $\mathfrak{h}_5$
of finite-difference (shift) operators can be formed:
\begin{alignat}{5}
 &[D_{\de_1, x} , X_{\de_1, x}]=1 ,&& [D_{\de_2, y} , X_{\de_2, y}]=1 , && [D_{\de_1, x} , X_{\de_2, y}]=0 , && [D_{\de_2, y} , X_{\de_1, x}]=0 ,&
\nonumber\\
 &[D_{\de_1, x}, D_{\de_2, y}]=0 ,\quad && [X_{\de_1, x}, X_{\de_2, y}]=0 , \quad && [D_{\de_1 (\de_2), x(y)}, I]=0 ,\quad && [X_{\de_1 (\de_2), x(y)}, I]=0 .& \label{h5-delta}
\end{alignat}
Evidently, the vacuum vector,
\[
 |0\rangle = 1 ,
\]
for any $(\de_1,\de_2) \in {\R}^2$.

This algebra acts on the rectangular uniform lattice with spacings $(\de_1,\de_2)$.
By identifying in~(\ref{hA2-pq}) and~(\ref{kA2-pq}) the variables $(p,q)$ with $(D_{\de}, X_{\de})$,
we arrive at the Hamiltonian and the integral of the polynomial integrable model
on the two-dimensional uniform lattice with spacings $(\de_1,\de_2)$,%
\begin{equation}
\label{h-delta}
 h^{(\de_1,\de_2)} = h_{A_2}(D_{\de_1, x}, X_{\de_1, x}, D_{\de_2, y}, X_{\de_2, y}) ,
\end{equation}
and
\begin{equation}
\label{k-delta}
 k^{(\de_1,\de_2)} = k_{A_2}(D_{\de_1, x}, X_{\de_1, x}, D_{\de_2, y}, X_{\de_2, y}) ,
\end{equation}
If parameter $-\nu=n/3$, $n=0,1,2,\dots $, the eigenvalue problems for the operators~(\ref{h-delta}) and~(\ref{k-delta}) have $\frac{(n+2)(n+1)}{2}$ common polynomial eigenfunctions
$\phi^{(h,k)}(x, y)$ in the form of triangular polynomials,
\[
 \langle x^{m_x} y^{m_y} \mid 0 \leq m_x + m_y \leq n\rangle .
\]
The first polynomial eigenfunctions for $n=0,1,2$ can be easily found by using the results collected in Examples~\ref{examples}.

\subsection{Exponential dilatation-invariant lattice}

Let us introduce the dilation operator,
\[
 T_{q} f(x) = f(q x),\qquad T_{q}=q^{A} ,\qquad A \equiv x \pa_x ,
\]
where $q \in {\C}$, and construct a canonical pair of dilatation operators
\begin{equation}
\label{q}
 D_{q} = x^{-1}\frac{T_{q} - 1}{q-1} ,\qquad X_{q} = \frac{A (q-1)}{T_{q} - 1} x ,
\end{equation}
see \cite{CT}, where $[D_{q} , X_{q}]=1$ for any $q$. It can be checked that their product
is $q$-independent,
\[
 X_{q} D_{q}=x \pa_x = A \qquad \mbox{and}\qquad D_{q} X_{q}=\pa_x x=A+1 .
\]
The operator $D_q$ is called the Jackson symbol (or the Jackson derivative).
Both operators~$X_{q}$,~$D_{q}$ are pseudodifferential operators with action
on monomials as follows:
\[
 D_{q} x^n = \{ n \}_q x^{n-1} ,\qquad X_q x^n = \frac{n+1}{\{ n+1\}_q} x^{n+1} ,
\]
where $\{ n \}_q=\frac{1 - q^n}{1 - q}$ is the so called $q$-number $n$.

By taking $D_{q}$, $X_{q}$~(\ref{q}) as basic elements it can be shown that algebra $\mathfrak{h}_5$
of discrete operators can be formed:
\begin{alignat*}{5}
 &[D_{q_1, x} , X_{q_1, x}]=1 ,\quad&& [D_{q_2, y} , X_{q_2, y}]=1 ,\quad&& [D_{q_1, x} , X_{q_2, y}]=0 ,\quad&& [D_{q_2, y} , X_{q_1, x}]=0 ,&
\\
 &[D_{q_1, x}, D_{q_2, y}]=0 ,\qquad&& [X_{q_1, x}, X_{q_2, y}]=0 ,\qquad&& [D_{q_1 (q_2), x(y)}, I]=0 ,\qquad&& [X_{q_1 (q_2), x(y)}, I]=0 ,&
\end{alignat*}
cf.~(\ref{h5-delta}). Evidently, the vacuum vector,
\[
 |0\rangle = 1
\]
for any $(q_1, q_2) \in {\R}^2$.

This algebra acts on the exponential lattice with spacings $(q_1,q_2)$.
By identifying in~(\ref{hA2-pq}) and~(\ref{kA2-pq}) the variables $(p,q)$ with $(D_{q}, X_{q})$
we arrive at the Hamiltonian and the integral of the polynomial integrable model
on the two-dimensional exponential lattice with spacings $(q_1,q_2)$,
\begin{equation}
\label{h-q}
 h^{(q_1,q_2)} = h_{A_2}(D_{q_1, x}, X_{q_1, x}, D_{q_2, y}, X_{q_2, y})
\end{equation}
and
\begin{equation}
\label{k-q}
 k^{(q_1,q_2)} = k_{A_2}(D_{q_1, x}, X_{q_1, x}, D_{q_2, y}, X_{q_2, y}) .
\end{equation}
If parameter $-\nu=n/3$, $n=0,1,2,\dots $, the eigenvalue problems for~(\ref{h-q}) and~(\ref{k-q}) have $\frac{(n+2)(n+1)}{2}$ common polynomial eigenfunctions
$\phi^{(h,k)}(x, y)$ in the form of triangular polynomials,
\[
 \langle x^{m_x} y^{m_y}\mid 0 \leq m_x + m_y \leq n\rangle .
\]
The first polynomial eigenfunctions for $n=0,1,2$ can be easily found by using the results collected in Examples~\ref{examples}.

\subsection{Mixed translation-invariant and dilatation-invariant lattice}

It is evident that one can construct the operators $h$, $k$ acting in $x$-direction on the uniform lattice and in $y$-direction on the exponential lattice and visa versa. Therefore,
there are two ways to realize it by taking
\begin{equation}
\label{xy-de,q}
 p_x=D_{\de_1, x} ,\qquad q_x=X_{\de_1, x} ,\qquad p_y=D_{q_1, y} ,\qquad q_y=X_{q_1, y} ,
\end{equation}
or
\begin{equation}
\label{xy-q,de}
 p_x=D_{q_2, x} ,\qquad q_x=X_{q_2, x} ,\qquad p_y=D_{\de_2, y} ,\qquad q_y=X_{\de_2, y} .
\end{equation}
In both cases the vacuum vector remains the same,
\[
 |0\rangle = 1 .
\]

\noindent
In a straightforward way one can build the Hamiltonian and the integral
\begin{equation}
\label{h-delta-q}
 h^{(\de_1,q_1)} = h_{A_2}(D_{\de_1, x}, X_{\de_1, x}, D_{q_1, y}, X_{q_1, y}) ,
\end{equation}
and
\begin{equation}
\label{k-delta-q}
 k^{(\de_1,q_1)} = k_{A_2}(D_{\de_1, x}, X_{\de_1, x}, D_{q_1, y}, X_{q_1, y}) ,
\end{equation}

\noindent
for~(\ref{xy-de,q}) and

\begin{equation}
\label{h-q-delta}
 h^{(q_2,\de_2)} = h_{A_2}(D_{q_2, x}, X_{q_2, x}, D_{\de_2, y}, X_{\de_2, y}) ,
\end{equation}
and
\begin{equation}
\label{k-q-delta}
 k^{(q_2,\de_2)} = k_{A_2}(D_{q_2, x}, X_{q_2, x}, D_{\de_2, y}, X_{\de_2, y}) ,
\end{equation}

\noindent
for~(\ref{xy-q,de}). In similar way as for~(\ref{hA2diff})--(\ref{kA2diff}),~(\ref{h-delta})--(\ref{k-delta}),~(\ref{h-q})--(\ref{k-q}),
if parameter $\nu=n/3$, ${n=0,1,2,\dots}$, the eigenvalue problems for
(\ref{h-delta-q})--(\ref{k-delta-q}) and~(\ref{h-q-delta})--(\ref{k-q-delta}) have $\frac{(n+2)(n+1)}{2}$ common polynomial eigenfunctions
$\phi^{(h,k)}(x, y)$ in the form of triangular polynomials,
\[
 \langle x^{m_x} y^{m_y} \mid 0 \leq m_x + m_y \leq n\rangle .
\]
The first polynomial eigenfunctions for $n=0,1,2$ can be easily found by using the results collected in Examples~\ref{examples}.

Remarkably, all these five integrable models~(\ref{hA2diff})--(\ref{kA2diff}),~(\ref{h-delta})--(\ref{k-delta}),~(\ref{h-q})--(\ref{k-q}) and~(\ref{h-delta-q})--(\ref{k-delta-q}),~(\ref{h-q-delta})--(\ref{k-q-delta}) are {\it isospectral}.

\section[gl(3)-polynomial integrable model in C\^{}2]{$\boldsymbol{\mathfrak{gl}(3)}$-polynomial integrable model in $\boldsymbol{{\bf \C^2}}$}\label{sec:4}

Introduce the five-dimensional Heisenberg algebra $\mathfrak{h}_5 \equiv \mathbb{H}_{5} = \big\{ \ag_1,\ag^\dag_1, \ag_2,\ag^\dag_2, I \big\}$ with commutator $\big[\ag_{i}, \ag^\dag_j\big] = \delta_{ij}I$, $i,j = 1,2$, $[\ag_{i}, \ag_j]=\big[\ag^\dag_{i}, \ag^\dag_j\big]=0$ and
$[\ag_{i}, I]=\big[\ag^\dag_{i}, I\big]=0$
by using a new, mathematics-oriented notations~\cite{TV}.
Its representation on the standard Hilbert space,
\[
L_2\bigl(\C^2,{\rm d}\mu_2\bigr)=L_2(\C,{\rm d}\mu) \otimes L_2(\C,{\rm d}\mu) ,
\]
with the Gaussian measure,
\begin{equation*}
{\rm d}\mu(z)=\pi^{-1}{\rm e}^{-z\cdot {\bar z}}{\rm d}v(z) ,
\end{equation*}
where ${\rm d}v(z)={\rm d}x{\rm d}y$ is the Euclidean volume measure on $\C=\R^{2}$, is given by two canonical pairs of raising and lowering operators related to $z=(z_1,z_2) \in \C^2$:
\begin{equation*}
 \ag^\dag_j = {\bar z}_j -\frac{\pa}{\pa {z}_j}, \qquad \ag_j =
 \frac{\pa}{\pa {\bar z}_j} ,
\end{equation*}
where $\ag^\dag_j$ is adjoint to $\ag_j$ with $\big[\ag_j,\ag^\dag_j\big]=I$, $j = 1,2$,
see \cite{TV} for details and $I$ is the identity operator. The vacuum vector
$|0\rangle$, defined by
\[
 \ag_1|0\rangle=0 ,\qquad \ag_2|0\rangle=0 ,
\]
is any two-dimensional analytic function.

Formally, by taking~(\ref{hA2-pq}) and~(\ref{kA2-pq}) one can build the Hamiltonian
\begin{equation*}
 h^{(\C^2)} = h_{{\mathcal A}_2}\bigl(\ag_1, \ag^\dag_1, \ag_2, \ag^\dag_2\bigr) ,
\end{equation*}
and the integral
\begin{equation*}
 k^{(\C^2)} = k_{{\mathcal A}_2}\bigl(\ag_1, \ag^\dag_1, \ag_2, \ag^\dag_2\bigr) .
\end{equation*}
It is evident that they continue to commute. This procedure can be considered as a complexification of the original polynomial model~(\ref{hA2diff}) and~(\ref{kA2diff}),
which emerged from the 3-body/${\mathcal A}_2$ elliptic Calogero model as its algebraic version.
Formally, the Hamiltonian is the sixth-order differential operator in $\frac{\pa}{\pa {z}}$, $\frac{\pa}{\pa {\bar z}}$.
Note that the first polynomial eigenfunctions of $h^{(\C^2)}$ for $n=0,1,2$ can be easily found by using the results collected in Examples~\ref{examples}.

\section[gl(3) algebra: artifacts]{$\boldsymbol{\mathfrak{gl}(3)}$ algebra: artifacts}\label{sec:5}

Long ago one of the authors (Alexander V.~Turbiner) discovered in the algebra $\mathfrak{gl}(3)$ with generators defined in~(\ref{gl3})
the existence of nine bilinear combinations in generators with unusual property:
all those bilinear combinations vanish if the representation of $\mathfrak{gl}(3)$ generators
by the first-order differential operators~(\ref{sl3do}) is taken!
The explicit form of the bilinear combinations is the following \cite{at:1994}:
\begin{alignat}{4}
&A_1 = J_{8}J_{5} - J_{7}J_{6} ,\qquad &&
A_2 = J_{8}J_{3} - J_{7}J_{4} ,\qquad &&
A_3 = J_{7}J_{2} + J_{5}J_{0} + J_{5} ,&
\nonumber
\\
&A_4 = J_{8}J_{1} + J_{4}J_{0} + J_{4} ,\qquad &&
A_5 = J_{7}J_{1} + J_{3}J_{0} + J_{3} ,\qquad &&
A_6 = J_{8}J_{2} + J_{6}J_{0} + J_{6} ,&
\nonumber
\\
&A_7 = J_{6}J_{3} - J_{5}J_{4} + J_{3} ,\qquad &&
A_8 = J_{6}J_{1} - J_{4}J_{2} ,\qquad &&
A_9 = J_{5}J_{1} - J_{3}J_{2} . &
\label{gl3artifacts}
\end{alignat}

\begin{Theorem}\label{t2}
For the $\mathfrak{gl}(3)$ generators, written in terms of the Heisenberg algebra $\mathfrak{h}_5$ generators~\eqref{gl3pq}, all nine artifacts~\eqref{gl3artifacts} vanish
\[
 A_{1,\dots,9}(p_{x,y}, q_{x,y}) = 0 .
\]
\end{Theorem}
This theorem can be proved by direct calculation.

It can be checked that the commutators of the artifacts are of the form
\[
 [A_i, A_J] = \sum c_{ij}^k(J) A_k ,
\]
where $c_{ij}^k(J)$ for $i \neq j$ are linear combinations of the $\mathfrak{gl}(3)$ generators. For example,
\[
 [A_1, A_2] = - (J_8 A_1 + J_7 A_2) ,
\]
see Appendix~\ref{app:c} for the explicit calculation.
Hence, $A_i$ do not span a Lie algebra. Interesting question to ask is what would happen
with the artifacts $A$'s if instead of $\mathfrak{gl}(3)$ generators the~$\mathfrak{gl}(3)$ Kac--Moody currents are taken. It will be addressed elsewhere.

\section[The Hamiltonian and the integral in gl(3)-algebra generators]{The Hamiltonian and the integral in $\boldsymbol{\mathfrak{gl}(3)}$-algebra generators}\label{sec:6}

\subsection{Hamiltonian}
By taking the Hamiltonian~(\ref{hA2-pq}), one can demonstrate that it can be rewritten in
the $\mathfrak{gl}(3)$ abstract generators, which obey formally the commutation relations~(\ref{gl3}),
\begin{align}
h_{{\mathcal A}_{2}}(J) ={}& 2 J_6 J_1-\frac{1}{3}J_5^2-J_1 J_0
+\mu(2J_8J_5+J_7 J_3-2J_7 J_0+3J_4^2+2J_7)
\nonumber
\\ &{}+
 \tau(4J_8 J_2+4J_7 J_1-J_6^2-J_3^2+5J_6+5J_3)
- 3\tau\mu J_8 J_4 - 3\mu^2J_8^2 - 3\tau^2J_4^2 ,\label{hA2-J}
\end{align}
hence, in extremely compact form; here $\mu$, $\tau$ are parameters and the dependence on
$\nu$ can be included into the representation (into the generators) and eventually
is absent!
Hence,~(\ref{hA2-J}) is two-parametric, bilinear element of the universal enveloping
algebra $U_{\mathfrak{gl}(3)}$. If $\mu=\tau=0$ the element $h_{{\mathcal A}_{2}}$~(\ref{hA2-J}) dramatically simplifies,
\begin{eqnarray*}
h^{(\mu=\tau=0)}_{{\mathcal A}_{2}}(J) = 2 J_6 J_1 - \tfrac{1}{3} J_5^2-J_1 J_0 .
\end{eqnarray*}
By substituting the generators $J_{0,1,5,6}$ in the form of differential operators
(\ref{sl3do}), one can see that this element corresponds to the 3-body/${\mathcal A}_2$ rational Calogero model (without harmonic oscillator term).
Non-surprisingly, the raising generators $J_{7,8}$ are absent in this case, as well as the generators $J_{4,3,2}$.

\subsection{Cubic integral}

In a similar way, as was done in order to construct~(\ref{hA2-J}), by taking the integral~(\ref{kA2-pq}) in the Fock space representation one can demonstrate
that it can be rewritten in the $\mathfrak{gl}(3)$ abstract generators which obey
the commutation relations
(\ref{gl3}),
\begin{align}
k_{{\mathcal A}_{2}}(J) = {}& -\frac{2}{9}J_6^2 J_2
+\frac{2}{9}J_6 J_5 J_1
+\frac{5}{9}J_6 J_2 J_0
-\frac{2}{27}J_5^3
+\frac{2}{9}J_5 J_1 J_0
+J_4J_1^2
\nonumber
\\ &{}
-\frac{2}{9}J_3^2J_2
-\frac{2}{9}J_2J_0^2
+\frac{2}{9}J_6J_2
+\frac{2}{9}J_5J_1
+\frac{2}{9}J_2J_0
\nonumber
\\ &{}
- \tau\biggl(
 \frac{8}{9}J_7J_6J_2
+\frac{8}{9}J_7J_5J_1
-\frac{8}{9}J_7J_2J_0
+\frac{2}{9}J_6J_6J_5
-\frac{2}{9}J_6J_5J_3
+\frac{2}{9}J_5J_3J_3
\nonumber
\\ &\quad{}
-2J_4J_3J_1
-3J_8J_1J_1
+\frac{2}{3}J_6J_5
+\frac{2}{3}J_5J_3
-\frac{16}{9}J_5J_0
-4J_4J_1\biggr)
\nonumber
\\ &{}
+ \tau^2\biggl(\frac{2}{3}J_6^2J_4
-\frac{2}{3}J_6J_4J_3
-\frac{8}{3}J_6J_4J_0
+\frac{2}{3}J_4J_3^2
-\frac{8}{3}J_4J_3J_0
+\frac{8}{3}J_4J_0^2
\nonumber
\\ &\quad{}
-\frac{4}{3}J_6J_4
-\frac{4}{3}J_4J_3
+ \frac{8}{3}J_4J_0
+ 2J_4\biggr)
+ 2\tau^3J_4^3
\nonumber
\\ &{}
- \mu\biggl(\frac{1}{3}J_7J_6J_5
+\frac{2}{3}J_7J_5J_3
-\frac{4}{3}J_7J_5J_0
+\frac{2}{3}J_6^2J_4
-\frac{2}{3}J_6J_4J_3
-\frac{8}{3}J_6J_4J_0
\nonumber
\\ &\quad{}
-\frac{1}{3}J_4J_3^2
+\frac{10}{3}J_4J_3J_0
-\frac{1}{3}J_4J_0^2
+\frac{4}{3}J_7J_5
-\frac{4}{3}J_6J_4
+\frac{5}{3}J_4J_3
-\frac{1}{3}J_4J_0\biggr)
\nonumber
\\ &{}
-\mu\tau\biggl(
4J_8J_0
-\frac{1}{3}J_8J_6^2
+\frac{28}{3}J_8J_6J_3
+\frac{4}{3}J_8J_6J_0
-\frac{7}{3}J_8J_3^2
\nonumber
\\ &\quad{}
+\frac{16}{3}J_8J_3J_0
-\frac{4}{3}J_8J_0^2
-10J_7J_6J_4
+3J_4^3
-J_8J_6
+7J_8J_3
-\frac{8}{3}J_8
\biggr)
\nonumber
\\ &{}
+ 3\mu\tau^2J_8J_4^2 - 3\mu^2\tau J_8^2J_4
\nonumber
\\ &{}
+ \mu^2\bigl(2J_8 J_7 J_6
+J_8 J_7 J_3
-2J_8 J_7 J_0
-6J_8 J_4^2
+4J_8 J_7\bigr)
-2\mu^3 J_8^3 ,\label{kA2-J}
\end{align}

\noindent
where $\mu$, $\tau$ are parameters, see~\eqref{wp-inv}, and the explicit dependence on $\nu$ is absent! Hence, it is two-parametric, trilinear element of the universal enveloping algebra $U_{\mathfrak{gl}(3)}$. If $\mu=\tau=0$, the element $k_{{\mathcal A}_{2}}$~(\ref{kA2-J}) dramatically simplifies,
\begin{gather*}
k^{(\mu=\tau=0)}_{{\mathcal A}_{2}}(J) = -\frac{2}{9}J_6^2 J_2 + \frac{2}{9}J_6 J_5 J_1 + \frac{5}{9}J_6 J_2 J_0
-\frac{2}{27}J_5^3 + \frac{2}{9}J_5 J_1 J_0
+J_4J_1^2 -
\frac{2}{9}J_3^2J_2
 \\ \hphantom{k^{(\mu=\tau=0)}_{{\mathcal A}_{2}}(J) =}{}
-\frac{2}{9}J_2J_0^2
+\frac{2}{9}J_6J_2
+\frac{2}{9}J_5J_1
+\frac{2}{9}J_2J_0 ,
\end{gather*}
it corresponds to the 3-body/${\mathcal A}_2$ rational integrable Calogero model (without harmonic oscillator term). Since this is the exactly-solvable problem, non-surprisingly, the raising generators $J_{7,8}$ are absent.

\subsection{Commutator}

By taking~(\ref{hA2-J}) and~(\ref{kA2-J}), one can make the extremely cumbersome (and very lengthy) calculation of their Lie bracket (commutator) by using the specially designed \textsc{Maple}~18 code.
An example of the code is presented in Appendix~\ref{app:c}.
It was the main goal of the master thesis of one of the authors (Miguel A.~Guadarrama-Ayala).
Eventually, it leads to the following statement:

\begin{Theorem}
The commutator of the expressions~\eqref{hA2-J} and~\eqref{kA2-J}
is the linear superposition of artifacts~\eqref{gl3artifacts},
\begin{equation}
\label{commutator}
[h_{{\mathcal A}_2}(J), k_{{\mathcal A}_2}(J)] = \sum_{i=1}^9 c_i(J) A_i ,
\end{equation}
for any values of parameters $\ta$, $\mu$, where $c_i(J)$ are some coefficient functions in $J$'s.
\end{Theorem}

\begin{proof} By direct calculation by using the specially designed \textsc{Maple}~18 code. It was carried out on a regular DELL desktop computer with 2.4~GhZ working frequency and 6~GB RAM memory.

Intuitively, this result~(\ref{commutator}) is evident: in the Fock space representation, where $h,k \in U_{\mathfrak{h}_5}$, the commutator should vanish, see Theorems~\ref{t1} and~\ref{t2}. Hence, the commutator should be a~(non)-linear combination of the artifacts. The fact that is a linear combination of the artifacts is non-trivial.

Alternative way to represent the commutator~(\ref{commutator}) is as follows
\begin{align}
[ h_{{\mathcal A}_2}(J),k_{{\mathcal A}_2}(J) ]={}&
D_1 + D_2\tau + D_3\mu + D_4 \tau^2 + D_5 \tau\mu
+ D_6 \mu^2
+ D_7 \tau^2\mu+ D_8 \tau\mu^2
\nonumber\\&{}
+ D_9 \mu^3 + D_{10} \tau^3\mu
+ D_{11} \tau^2\mu^2 + D_{12} \tau\mu^3 ,
\nonumber
\end{align}
where for the coefficients $D(J,A)$ are presented in Appendix~\ref{app:b}.
\end{proof}

\section[G\_2 elliptic 3-body problem]{$\boldsymbol{G_2}$ elliptic 3-body problem}\label{sec:7}

By adding the 3-body interaction potential to the 3-body elliptic Calogero Hamiltonian~(\ref{OPHam2}), we arrive at the 3-body Wolfes elliptic Hamiltonian
in $(y_1, y_2)$-coordinates~(\ref{y}),
\begin{gather}
 {\mathcal H}_{{G}_2} = - \frac{1}{3} \biggl(\frac{\pa^{2}}{\pa y_1^{2}}+
 \frac{\pa^{2}}{\pa y_2^{2}} - \frac{\pa^{2}}{\pa y_1 \pa y_2}\biggr)
 \nonumber\\ \hphantom{{\mathcal H}_{{G}_2} =}{}
 +(\nu-\la)(\nu-\la-1) (\wp (y_1-y_2) + \wp (2y_1+y_2) + \wp (y_1+2y_2))
\nonumber\\ \hphantom{{\mathcal H}_{{G}_2} =}{}
 + \la(3\la-1) (\wp (y_1) + \wp (y_2) + \wp (y_1+y_2)) ,\label{OPHam2G}
\end{gather}
which is also called the $G_2$ elliptic Hamiltonian in the Hamiltonian reduction nomenclature~\cite{Olshanetsky:1983}. It is characterized by two coupling constants which
can be parameterized conveniently as $\ka \equiv (\nu-\la)(\nu-\la-1)$ and
$\ka_2 \equiv \la(3\la-1)$. If $\ka_2=0$ (or $\la=0, 1/3$), we return at the ${\mathcal A}_2$ elliptic model. It was shown in \cite{ST:2015} that by making the gauge rotation and changing variables
to $\bigl(u=x, v=y^2\bigr)$, see~(\ref{trans}), the Hamiltonian~(\ref{OPHam2G}) appears in the form of the algebraic operator~$h_{G_2}$~-- the second-order differential operator with polynomial coefficients,
\begin{gather}
h_{G_2}(u,v) = \bigl(u+3 \tau u^2+3 \mu u^3 +3 \bigl(\mu-\ta^2\bigr) v-3 \mu \ta u v
 -3\mu^2 u^2 v \bigr)\frac{\pa^2}{\pa u^2}
\nonumber\\ \hphantom{h_{G_2}(u,v) =}{}
+2 v \bigl(3+8\ta u+7\mu u^2-3\mu \ta v -6\mu^2 u v \bigr)\frac{\pa^2}{\pa u \pa v}
\nonumber\\ \hphantom{h_{G_2}(u,v) =}{}
+ 4v \biggl(-\frac{u^2}{3}+3 \tau v+4 \mu u v-3\mu^2 v^2 \biggr)\frac{\pa^2}{\pa v^2}
\nonumber \\ \hphantom{h_{G_2}(u,v) =}{}
+ (1+3\nu) \bigl(1+4 \ta u+5 \mu u^2-3\mu \ta v-6\mu^2 u v \bigr) \frac{\pa}{\pa u}
\nonumber \\ \hphantom{h_{G_2}(u,v) =}{}
+ 2\biggl(-\frac{u^2}{3} + \ta (7+12\nu) v + 2 \mu (5+9\nu) u v
 - 9 \mu^2 (1+2\nu) v^2 \biggr)\frac{\pa}{\pa v}
\nonumber \\ \hphantom{h_{G_2}(u,v) =}{}
+ 3\nu(1+3\nu) \mu (2 u -3\mu v )
\nonumber \\ \hphantom{h_{G_2}(u,v) =}{}
+\la \biggl( 6 \bigl(1+2 \tau u+\mu u^2\bigr) \frac{\pa}{\pa u} +
4 \bigl(-u^2+3 \tau v+3 \mu u v\bigr)\frac{\pa}{\pa v} +18 \nu \mu u \biggr) .
\label{hG2}
\end{gather}
If $-3\nu=n$, where $n=0,1,2,3,\dots$, this operator has a finite-dimensional triangular invariant subspace,
\[
 {\mathcal P}_n = \langle u^{n_u} v^{n_v}\mid 0 \leq n_u+2n_v \leq n\rangle.
\]
This space coincides with finite-dimensional representation space of the algebra $\mathfrak{g}^{(2)}$: infinite-dimensional, eleven-generated algebra of differential operators, see \cite{ST:2015} and references therein, where it acts irreducibly. It implies that $h_{G_2}$ can be rewritten in terms $\mathfrak{g}^{(2)}$ generators (the Burnside theorem) \cite{ST:2015}.

After extremely tedious (and very lengthy) symbolic calculations by using the \textsc{Maple}~18 code, see Appendix~\ref{app:c} for an example, it can be shown that the existence of a differential operator~$k_m(u,v)$ of degree five such that the operator
\[
 k_{\rm G_2} = k_{{\mathcal A}_2}^2(u,v) + \la k_m(u,v; \la) ,
\]
commutes with the $G_2$ elliptic Hamiltonian $h_{G_2}$~(\ref{hG2}); $k_m$ has the form of polynomial in $\la$ of finite degree. Note that in the particular case of the $G_2$ rational Hamiltonian (see~(\ref{hG2}) at $\mu=\tau=0$), this operator was calculated in \cite{TTW:2009} (where it corresponded to the case $k=3$) in slightly different variables other than $u$, $v$: it is a polynomial in $\la$ of degree four. In general, this operator will be presented in its explicit form elsewhere. So far this operator is unknown in the explicit form.

By taking the 5-dimensional Heisenberg algebra $\mathfrak{h}_5$ spanned by the generators~$p_u$,~$p_v$,~$q_u$,~$q_v$,~$I$, see~(\ref{h5}), one can form the following second degree polynomial in~$p_u$, $p_v$:
\begin{gather}
h_{G_2}(p_u, p_v, q_u, q_v) = \bigl(q_u+3 \tau q_u^2+3 \mu q_u^3 +
3 \bigl(\mu-\ta^2\bigr) q_v-3 \mu \ta q_u q_v -3\mu^2 q_u^2 q_v \bigr) p_u^2
\nonumber \\ \hphantom{h_{G_2}(p_u, p_v, q_u, q_v) =}{}
+ 2 q_v \bigl(3+8\ta q_u+7\mu q_u^2-3\mu \ta q_v -6\mu^2 q_u q_v \bigr)p_u p_v
\nonumber \\ \hphantom{h_{G_2}(p_u, p_v, q_u, q_v) =}{}
+ 4q_v \biggl(-\frac{q_u^2}{3}+3 \tau q_v+4 \mu q_u q_v-3\mu^2 q_v^2 \biggr) p_v^2
\nonumber \\ \hphantom{h_{G_2}(p_u, p_v, q_u, q_v) =}{}
+ (1+3\nu) \bigl(1+4 \ta q_u+5 \mu q_u^2-3\mu \ta q_v-6\mu^2 q_u q_v \bigr) p_u
\nonumber \\ \hphantom{h_{G_2}(p_u, p_v, q_u, q_v) =}{}
+ 2\biggl(-\frac{q_u^2}{3} + \ta (7+12\nu) q_v + 2 \mu (5+9\nu) q_u q_v
 - 9 \mu^2 (1+2\nu) q_v^2 \biggr) p_v
\nonumber \\ \hphantom{h_{G_2}(p_u, p_v, q_u, q_v) =}{}
+ 3\nu(1+3\nu) \mu \bigl(2 q_u -3\mu q_v \bigr)
\label{hG2pq}\\ \hphantom{h_{G_2}(p_u, p_v, q_u, q_v) =}{}
+ \la \bigl( 6 \bigl(1+2 \tau q_u+\mu q_u^2\bigr) p_u
+4 \bigl(-q_u^2+3 \tau q_v+3 \mu q_u q_v\bigr) p_v +18 \nu \mu q_u \bigr) .
\nonumber
\end{gather}
It is easy to check that if $(p,q)$-variables are taken in the coordinate-momentum representation,
\[
 p_u = \frac{\pa}{\pa u} ,\qquad p_v = \frac{\pa}{\pa v} ,\qquad q_u = u ,
 \qquad q_v = v ,
\]
cf.~(\ref{h3-pq}), the expression~(\ref{hG2pq}) is reduced to the operator~(\ref{hG2}).
The operator $h_{G_2}(p_u, p_v, q_u, q_v)$ represents the $G_2$ elliptic model
in the Fock space.

By substituting into~(\ref{hG2pq}) the representations~(\ref{h3-D-delta})--(\ref{h3-TV}), we will arrive at the $G_2$ elliptic lattice Hamiltonians defined on uniform-uniform, uniform-exponential, exponential-uniform, exponential-exponential lattices
in $(u, v)$ space as well as the complexified $G_2$ elliptic Hamiltonian in the algebraic form.

\section{Conclusions}

In this paper, a polynomial integrable system, associated with the algebra $U_{\mathfrak{h}_5}$ and inspired
by the algebraic representation of the ${\mathcal A}_2$ elliptic model in Fock space is defined.
It has the form of a second degree polynomial in $p_i$, $i=1,2$,
\begin{equation}
\label{h}
 h_{{\mathcal A}_2} = c^{(2)}_{ij} p_{i} p_{j} + c^{(1)}_{i} p_{i} + c^{(0)} ,
\end{equation}
for the Hamiltonian and a 3rd degree polynomial in $p_i$, $i=1,2$,
\begin{equation}
\label{k}
 k_{{\mathcal A}_2} = d^{(3)}_{ijk} p_{i} p_{j} p_{k} + d^{(2)}_{ij} p_{i} p_{j} + d^{(1)}_{i} p_{i} + d^{(0)} ,
\end{equation}
for the integral, where the coefficients $\{c\}$ and $\{d\}$ are polynomials in $q$ of a finite degrees, while~$(p_i, q_i)$ form a canonical pair. Overall, the operators $h_{{\mathcal A}_2}$ and $k_{{\mathcal A}_2}$ depend on three free parameters $\mu$, $\tau$, $\nu$. Remarkably, both operators $h_{{\mathcal A}_2}$ and $k_{{\mathcal A}_2}$ can be rewritten in terms of the $\mathfrak{sl}(3)$ generators $J_{1,2,\dots,8}$ and they can be embedded into the $U_{\mathfrak{h}_5}$ algebra in the $(-3\nu,0)$ representation~(\ref{gl3pq}). Hence, $\nu$ corresponds to the mark of the representation.

It can be conjectured that

\begin{Conjecture}
Up to canonical transformation
\[
 p \rar p + f(q) ,\qquad q \rar q ,
\]
there are no other non-trivial commuting operators in the $U_{\mathfrak{h}_5}$ algebra of degree 2 and 3
in $p$ other than $h$~(\ref{h}) and $k$~(\ref{k}).
\end{Conjecture}

The operators~$h$ and~$k$ can be rewritten in terms of the abstract $\mathfrak{gl}(3)$ generators
which obey the commutation relations~(\ref{gl3}) and which give a non-vanishing commutator $[h,k]$. However, once the $\mathfrak{gl}(3)$ generators are taken in the concrete representation~(\ref{gl3pq}) the operators~$h$ and~$k$ becomes $h_{{\mathcal A}_2}$~(\ref{hA2-pq}) and $k_{{\mathcal A}_2}$~(\ref{kA2-pq}), respectively, and their commutator $[h_{{\mathcal A}_2},k_{{\mathcal A}_2}]=0$. The remarkable property of the commutator $[h,k]$ is that it can be written as a linear superposition of the artifacts $A_{1,2,\dots,9}$. We doubt there exist other elements of the universal enveloping algebra~$U_{\mathfrak{gl}(3)}$ (up to automorphisms) with such a property.

Different realizations of $(p_i, q_i)$, $i=1,2$ by differential operators, finite-difference
operators, discrete operators, or the operators in $z$, ${\bar z}$ variables lead to a variety of concrete isospectral quantum integrable polynomial systems in two continuous variables, on two-dimensional uniform, exponential lattices or mixed ones, and on the ${\C}^2$ complex space. All these integrable models depend on the continuous parameter $\nu$. If this parameter takes certain discrete values, all above-mentioned integrable systems become quasi-exactly-solvable ones admitting a finite number of polynomial eigenfunctions in the form of triangular polynomials.

\appendix

\section[gl(3) algebra]{$\boldsymbol{\mathfrak{gl}(3)}$ algebra}\label{app:a}

The algebra $\mathfrak{gl}(3)$ is defined by nine generators $J_i$, $i=0,1,2,\dots, 8 $, which obey the following commutation relations:
\begin{alignat}{5}
& [ J_0, J_1] = J_1, \qquad && [ J_0, J_2] = J_{2} ,\qquad && [ J_0, J_3] = 0,\qquad && [ J_0, J_4] = 0 , &
\nonumber\\
& [ J_0, J_5] = 0,\qquad && [ J_0, J_6] = 0 ,\qquad && [ J_0, J_7] = -J_7,\qquad && [ J_0, J_8] = -J_8 , &
\nonumber\\
& [ J_1, J_2] = 0,\qquad && [ J_1, J_3] = J_{1} ,\qquad && [ J_1, J_4] = 0,\qquad && [ J_1, J_5] = J_{2} , &
\nonumber\\
& [ J_1, J_6] = 0,\qquad && [ J_1, J_7] = J_{3} - J_{0} ,\qquad && [ J_1, J_8] = J_{4} , &
\nonumber\\
& [ J_2, J_3] = 0,\qquad && [ J_2, J_4] = J_{1},\qquad && [ J_2, J_5] = 0,\qquad && [ J_2, J_6] = J_{2} , &
\nonumber\\
& [ J_2, J_7] = J_{5},\qquad && [ J_2, J_8] = J_{6} - J_{0} , &
\nonumber\\
& [ J_3, J_4] = -J_{4},\qquad && [ J_3, J_5] = J_{5},\qquad && [ J_3, J_6] = 0,\qquad && [ J_3, J_7] = J_{7}, &
\nonumber\\
& [ J_3, J_8] = 0 , &
\nonumber\\
&[ J_4, J_5] = - J_{3}+ J_{6}, \qquad && [ J_4, J_6] = - J_{4}, \qquad && [ J_4, J_7] = J_{8}, \qquad && [ J_4, J_8] = 0 , &
\nonumber\\
& [ J_5, J_6] = J_{5}, \qquad && [ J_5, J_7] = 0, \qquad && [ J_5, J_8] = J_{7} , &
\nonumber\\
& [ J_6, J_7] = 0, \qquad && [ J_6, J_8] = J_{8} ,\qquad && [ J_7, J_8] = 0 . & \label{gl3}
\end{alignat}

\subsection{Structure constants}

The commutation relations~(\ref{gl3}) of the $\mathfrak{gl}(3)$ algebra can be represented as
\[
 [J_i,J_j] = c_{ij}^k J_k , \qquad i,j,k=0,\dots, 8,
\]
where $c_{ij}^k$ are the structure constants. The non-vanishing structure constants are:
\begin{alignat*}{6}
 & c_{01}^1 = 1,\qquad && c_{02}^2=1,\qquad && c_{07}^7=-1,\qquad && c_{08}^8=-1 , &
\\
 & c_{13}^1 = 1,\qquad && c_{15}^2=1,\qquad && c_{17}^3=1,\qquad && c_{17}^0=-1, \qquad && c_{18}^4=1 , &
\\
 & c_{24}^1 = 1,\qquad && c_{26}^2=1,\qquad && c_{27}^5=1,\qquad && c_{28}^6=1, \qquad && c_{28}^0=-1 , &
\\
 & c_{34}^4 = -1,\qquad && c_{35}^5=1,\qquad && c_{37}^7=1 , &
\\
 & c_{45}^3 = -1,\qquad && c_{45}^6=1,\qquad && c_{46}^4=-1, \qquad && c_{47}^8=1 , &
\\
 & c_{56}^5 = 1,\qquad && c_{58}^7=1 , \qquad && c_{68}^8 = 1 . &
\end{alignat*}

\subsection[Representation of gl(3) algebra in differential operators]{Representation of $\boldsymbol{\mathfrak{gl}(3)}$ algebra in differential operators}

The algebra $\mathfrak{gl}(3)$ with commutation relations~(\ref{gl3}) can be realized by the first-order
differential operators in two variables,
\begin{gather}
J_1=\frac{\pa}{\pa x}, \qquad J_2=\frac{\pa}{\pa y}, \qquad J_3=x \frac{\pa}{\pa x}, \qquad
 J_4=y\frac{\pa}{\pa x},\qquad J_5=x\frac{\pa}{\pa y}, \qquad J_6=y \frac{\pa}{\pa y},
\nonumber
\\
J_7=x \biggl(x \frac{\pa}{\pa x}+ y \frac{\pa}{\pa y}+3\nu\biggr),\qquad
J_8=y \biggl(x \frac{\pa}{\pa x}+ y \frac{\pa}{\pa y}+3\nu\biggr), \label{sl3do}
\end{gather}
and
\[
 -J_0 = x \frac{\pa}{\pa x} + y \frac{\pa}{\pa y} + 3\nu = J_3+J_6+3\nu ,
\]
where $\nu$ is parameter.
It corresponds to the irreducible representation of the spin $(-3\nu,0)$. If $-3\nu=n$ is integer,
the finite-dimensional representation space which is spanned by triangular polynomials,
\[
 {\mathcal P}_n = \langle x^p y^q \mid 0 \leq (p+q) \leq n\rangle,
\]
occurs.

\subsection[Representation of gl(3) in (p,q) space]{Representation of $\boldsymbol{\mathfrak{gl}(3)}$ in $\boldsymbol{(p,q)}$ space}\label{app:a3}

Let us take 5-dimensional Heisenberg algebra $\mathfrak{h}_5$ spanned by the generators~$p_x$,~$p_y$,~$q_x$,~$q_x$,~$I$, which satisfy the commutation relations,
\begin{alignat*}{5}
& [p_x, q_x]=1 ,\qquad && [p_y, q_y]=1 ,\qquad && [p_x, q_y]=0 ,\qquad && [p_y, q_x]=0 , &
\\
& [p_x, p_y]=0 ,\qquad && [q_x, q_y]=0 ,\qquad && [p_{x,y}, I]=0 ,\qquad && [q_{x,y}, I]=0. &
\end{alignat*}
Define its universal enveloping algebra $U_{\mathfrak{h}_5}$ as the algebra of all ordered
monomials $\big\{ q_x^{i_x} q_y^{i_y} p_x^{j_x} p_y^{j_y}\big\}$. It is evident that the algebra
$\mathfrak{gl}(3)$ realized as
\begin{gather}
J_1=p_x , \qquad J_2=p_y , \qquad J_3=q_x p_x , \qquad J_4=q_y p_x , \qquad J_5= q_x p_y , \qquad J_6=q_y p_y ,\nonumber
\\
 J_7=q_x (q_x p_x + q_y p_y + 3\nu) ,\qquad
J_8=q_y (q_x p_x + q_y p_y + 3\nu) ,\label{gl3pq}
\end{gather}
and
\[
 -J_0=q_x p_x + q_y p_y + 3\nu = J_3 + J_6 + 3\nu ,
\]
is embedded into the universal enveloping algebra $U_{\mathfrak{h}_5}$.

Let us enlist four realizations of the commutation relation $[p_x, q_x]=1$:
\begin{itemize}\itemsep=0pt
 \item Continuous
\begin{equation}
\label{h3-pq}
 p_x = \frac{\pa}{\pa x} \equiv \pa_x ,\qquad q_x=x .
\end{equation}
It is well known, the so-called coordinate-momentum representation of the $\mathfrak{h}_3$ Heisenberg
algebra.

 \item On uniform lattice
\begin{equation}
\label{h3-D-delta}
 p_x = {\mathcal D}_{\delta} ,\qquad q_x = X_{\delta} ,
\end{equation}
with gap $\de$, where ${\mathcal D}_{\delta}$ is the Norlund derivative \cite{Turbiner:1995},
it is the basis for the so-called umbral calculus.

 \item On exponential lattice
\begin{equation*}
 p_x = {\mathcal D}_{q} ,\qquad q_x = X_{q} ,
\end{equation*}
where ${\mathcal D}_{q}$ is the Jackson derivative, $q$ has the meaning of the exponential spacing. It is described in details in \cite{CT}.

 \item Complex representation on $\C$
\begin{equation}
\label{h3-TV}
 \ag = \frac{\pa}{\pa {\bar z}} , \qquad \ag^\dag = -\frac{\pa}{\pa z} +
 {\bar z},
\end{equation}
see \cite{TV} and references therein.

\end{itemize}

\section{Coefficients in the commutator~(\ref{commutator})}\label{app:b}

The commutator between $h_{A_2}$ and $k_{A_2}$ can be written as the polynomial in parameters $\tau$, $\mu$,
\begin{align*}
[ h_{A_2}(J),k_{A_2}(J) ] ={}&
 D_1 + D_2\tau + D_3\mu + D_4 \tau^2 + D_5 \tau\mu + D_6 \mu^2 + D_7 \tau^2\mu
\\&{}
+ D_8 \tau\mu^2 + D_9 \mu^3 + D_{10} \tau^3\mu
+ D_{11} \tau^2\mu^2 + D_{12} \tau\mu^3 ,
\end{align*}
where the coefficients $D_{1,\dots, 12}$ are presented by superposition of the ordered polynomials
in $\mathfrak{gl}(3)$-generators $J_{0,1,\dots, 8}$ multiplied by the artifacts $A_{1,\dots, 9}$ of the $\mathfrak{gl}(3)$ algebra,
\begin{gather*}
D_{1} =
 - \frac{2}{9} ( 8 J_4J_2 + 3 J_3J_1 ) A_9
 - \frac{2}{9} (8 J_5J_1 - 8 J_3J_2 - 11 J_2J_0 ) A_8
\\ \hphantom{D_{1} =}{}
 - \frac{4}{3}J_2J_1 A_7 - \frac{22}{9}J_2J_1 A_6 + \frac{4}{9}J_2J_1 A_5
 + \frac{22}{9}J_2^2 A_4 - \frac{4}{9}J_1^2A_3 ,
\\ 
 D_2 =
 \frac{2}{9} \bigl(- 6 J_6^2 - 6J_5J_4 + 3J_3J_0 + 4 J_0^2 - 8J_6 + 3J_3 + 10 J_0
 - 14 \bigr) A_9
\\ \hphantom{D_2 =}{}
+ \frac{8}{9} (3J_6J_5 + 9J_4J_1 + 4J_5 ) A_8
 - \frac{2}{9} (12 J_5J_1 - 13 J_2J_0 ) A_7
 - \frac{28}{9}J_6J_2 A_5 + \frac{28}{9}J_6J_1 A_3 ,
\\ 
 D_3 =
 \frac{2}{9} \bigl(2J_8J_5 - 4 J_7J_3 + 3J_7J_0 - 36 J_4^2 + 4 J_7 \bigr) A_9
\\ \hphantom{D_3 =}{}
 + \frac{1}{3} (2 J_8J_1 - 7 J_7J_5 + 24 J_4J_3 + 30 J_4 ) A_8
\\ \hphantom{D_3 =}{}
 + \frac{1}{9} ( 5 J_7J_2 + 12 J_6J_5 - 12 J_5J_3 + 36 J_5J_0 - 10 J_5 ) A_7
\\ \hphantom{D_3 =}{}
 - \frac{4}{9} (3 J_5J_0 - 4 J_5 ) A_6
 - \frac{1}{9} (36 J_6J_5 - 16 J_5J_3 + 12 J_5J_0 + 63 J_4J_1 ) A_5
\\ \hphantom{D_3 =}{}
 + \frac{1}{3} (- 8 J_6J_1 - 10 J_4J_2 + 3 J_3J_1 + 6 J_1J_0 + 17 J_1 ) A_4
\\ \hphantom{D_3 =}{}
 + \frac{1}{9} \bigl( 4 J_6^2 - J_6J_0 - 4 J_5J_4 - 19 J_6 + 8 J_0 - 12 \bigr) A_3 +
 \frac{4}{3}J_5J_2 A_2 + \frac{2}{3}J_6J_2 A_1 ,
\\ 
 D_4 =
 \frac{8}{3} (3J_4J_3-2J_4J_0 )A_8 - 4J_4J_1A_7 - 10J_4J_1A_6 + 10J_4J_2 A_4 ,
\\ 
 D_5 =
 \frac{1}{3} (9 J_8J_6 + 48 J_8J_3 + 14 J_7J_4 + 71 J_8 ) A_8
 - \frac{2}{3} (2 J_7J_5 - 3 J_4J_3 + 20 J_4J_0 ) A_7
\\ \hphantom{D_5 =}{}
- \frac{2}{3} (16 J_8J_1 - 23 J_4J_3 ) A_6
 + \frac{1}{6} ( 83 J_8J_1 - 78 J_4J_3 + 219 J_4J_0 + 242 J_4 ) A_5
\\ \hphantom{D_5 =}{}
+ \frac{1}{6} \bigl( 64 J_8J_2 - 83 J_7J_1 - 124 J_6^2 + 34 J_6J_0 - 40 J_5J_4 + 50 J_3^2
\\ \hphantom{D_5 =+ \frac{1}{6} \bigl(}{}
 - 229 J_3J_0 + 54 J_6 + 32 J_0^2 - 297 J_0 - 66 \bigr) A_4
\\ \hphantom{D_5 =}{}
- \frac{2}{3} \bigl(41 J_6J_1 - 13 J_5^2 + 7 J_4J_2 \bigr) A_2
 - \frac{2}{3} (9 J_6J_5 + 4 J_5J_0 - J_5 ) A_1 ,
\\ 
D_{6} =
\frac{26}{3}J_8^2 A_9
 + \frac{2}{3} (3 J_8J_6 + 3 J_8J_3 - 26 J_8J_0 - J_8 ) A_7
\\ \hphantom{D_{6} =}{}
 - 6 (J_8 J_6 - J_8J_3 - J_8 J_0 ) A_6 -
 \frac{1}{3} (7 J_8J_3 + 10 J_8J_0 + 20 J_8 - 19 J_7 J_4 ) A_5
\\ \hphantom{D_{6} =}{}
 + \frac{1}{3} \bigl(36 J_7J_6 - 19 J_7J_0 - 90 J_4^2 + 21 J_7 \bigr) A_4
\\ \hphantom{D_{6} =}{}
 + \frac{1}{3} \bigl(19 J_7J_1 - 8 J_6^2 - 4 J_6J_3 + 50 J_6J_0 - 6 J_5J_4 +
 J_3^2 + 54 J_6 + 20 J_3 + 50 \bigr) A_2
\\ \hphantom{D_{6} =}{}
 - (3 J_8J_1 - J_7J_5 ) A_1 ,
\\ 
D_{7} =
2 (- 9 J_8J_6 + 4 J_8J_3 - 3 J_8 ) A_7 - 8 J_7J_4 A_6 +
4 (7J_8J_6 - 2 J_8J_3 + 4 J_8J_0 + 2 J_7J_4 ) A_5
\\ \hphantom{D_{7} =}{}
+ 4 \bigl( 2 J_8J_5 - 9 J_7J_6 - 4 J_7J_0 + 6 J_4^2 + 5 J_7 \bigr) A_4
 + 8 J_8J_4 A_3
\\ \hphantom{D_{7} =}{}
+ 2 (4 J_7J_1 - 2 J_6J_3 - 23 J_6 + 4J_0 + 6 J_3 + 23 ) A_2 -
 2 (4 J_8J_1 - 9 J_6J_4 ) A_1 ,
\\ 
D_8 =
-6J_8J_4 A_4 + \bigl(75J_4^2-27J_7J_6-2J_7J_3+4J_7J_0 \bigr) A_2
\\ \hphantom{D_8 =}{}
+ (15J_8J_6-16J_8J_3+20J_8J_0+25J_8 ) A_1 ,
\\ 
D_9 = - 18J_8^2 A_4 + 18J_8J_4 A_2 - 12J_8J_7 A_1 ,
\\ 
D_{10} = -66J_4^2 A_2 ,\qquad
D_{11} = -48J_8J_4 A_2 ,\qquad
D_{12} = -30J_8^2 A_2 .
\end{gather*}

\section[Maple code examples]{Maple code examples}\label{app:c}

In this appendix, we will provide an example of \textsc{Maple} code which carry out two test calculations:
the commutator $[J_4 J_1, J_5 J_1]$ and the commutator of two artifacts $[A_1, A_2]$.
Ultimate goal is to get a combination of lexicographically ordered monomials.

First upload \textsc{Maple} packages: \textsc{Physics}, \textsc{Library} and \textsc{StringTools}:
\begin{verbatim}
> with(Physics);
 with(Library);
 with(StringTools);
\end{verbatim}
Here we define the $\mathfrak{gl}(3)$ commutation relations among the generators $j0, j1, \dots, j8$
\begin{verbatim}
> Setup(noncommutativeprefix = j,
%Commutator(j1, j0) = -j1,
%Commutator(j2, j0) = -j2,
%Commutator(j2, j1) = 0,
%Commutator(j3, j0) = 0,
%Commutator(j3, j1) = -j1,
%Commutator(j3, j2) = 0,
%Commutator(j2, j3) = 0,
%Commutator(j4, j0) = 0,
%Commutator(j4, j1) = 0,
%Commutator(j4, j2) = -j1,
%Commutator(j4, j3) = j4,
%Commutator(j5, j0) = 0,
%Commutator(j5, j1) = -j2,
%Commutator(j5, j2) = 0,
%Commutator(j5, j3) = -j5,
%Commutator(j5, j4) = j3-j6,
%Commutator(j5, j6) = j5,
%Commutator(j5, j8) = j7,
%Commutator(j6, j0) = 0,
%Commutator(j6, j1) = 0,
%Commutator(j6, j2) = -j2,
%Commutator(j6, j3) = 0,
%Commutator(j6, j4) = j4,
%Commutator(j6, j5) = -j5,
%Commutator(j6, j7) = 0,
%Commutator(j7, j0) = j7,
%Commutator(j7, j1) = j0-j3,
%Commutator(j7, j2) = -j5,
%Commutator(j7, j3) = -j7,
%Commutator(j7, j4) = -j8,
%Commutator(j7, j5) = 0,
%Commutator(j8, j0) = j8,
%Commutator(j8, j1) = -j4,
%Commutator(j8, j2) = j0-j6,
%Commutator(j8, j3) = 0,
%Commutator(j8, j4) = 0,
%Commutator(j8, j5) = -j7,
%Commutator(j8, j6) = -j8,
%Commutator(j8, j7) = 0);
\end{verbatim}\vspace{-5mm}
\begin{gather*}
{\color{blue}[{\it algebrarules} = \{[j1, j0]_{-} = -j1, [j2, j0]_{-} = -j2, [j2, j1]_{-} = 0, [j3, j0]_{-} = 0,} \\
{\color{blue}[j3, j1]_{-} = -j1,[j3, j2]_{-} = 0, [j4, j0]_{-} = 0, [j4, j1]_{-} = 0, [j4, j2]_{-} = -j1, [j4, j3]_{-} = j4,}\\
{\color{blue}[j5, j0]_{-} = 0,
[j5, j1]_{-} = -j2, [j5, j2]_{-} = 0, [j5, j3]_{-} = -j5, [j5, j4]_{-} = j3-j6,}\\
{\color{blue}[j6, j0]_{-} = 0,
 [j6, j1]_{-} = 0, [j6, j2]_{-} = -j2, [j6, j3]_{-} = 0, [j6, j4]_{-} = j4, [j6, j5]_{-} = -j5,}
\\
{\color{blue}[j6, j7]_{-} = 0, [j7, j0]_{-} = j7, [j7, j1]_{-} = j0-j3, [j7, j2]_{-} = -j5, [j7, j3]_{-} = -j7,}
\\
{\color{blue}[j7, j4]_{-} = -j8, [j7, j5]_{-} = 0, [j8, j0]_{-} = j8, [j8, j1]_{-} = -j4, [j8, j2]_{-} = j0-j6,}
\\
{\color{blue}[j8, j3]_{-} = 0,
 [j8, j4]_{-} = 0, [j8, j5]_{-} = -j7, [j8, j6]_{-} = -j8, [j8, j7]_{-} = 0\},}
 \\
{\color{blue}noncommutativeprefix = \{j\}]}
 \end{gather*}
Here we define the lexicographic ordering of the $j$-generators via the list ``jorder''
\begin{verbatim}
> jorder := [j8, j7, j6, j5, j4, j3, j2, j1, j0]:
\end{verbatim}

Here we define the function OR which orders the $J$-generators in every term of an expression
in the lexicographic order defined above.

The function OR uses the command \textsc{SortProducts} which is a command of the \textsc{Physics Library}
which acts on a \textsc{Maple} expression containing noncommutative products. It uses the list ``jorder''
which defines the ordering of the $J$-generators. The option ``usecommutator'' indicates that
commutators defined previously.
\begin{verbatim}
> OR := proc (x)
Simplify(SortProducts(Expand(x), jorder, usecommutator))
end proc;
\end{verbatim}
Examples: product and commutator of two monomials
\begin{verbatim}
> monomial1 := j4 * j1:
> monomial2 := j5 * j1:
\end{verbatim}
Products of monomials
\begin{verbatim}
> productmonomials12 := monomial1 * monomial2;
\end{verbatim}\vspace{-5mm}
\[
{\color{blue}{\it productmonomials12:=} j4 j1 j5 j1}
\]
\begin{verbatim}
> productmonomials21 := monomial2 * monomial1;
\end{verbatim}\vspace{-5mm}
\[
{\color{blue}{\it productmonomials21:=} j5 j1 j4 j1}
\]
We apply the function OR to reorder the product of monomials in lexicographic order
\begin{verbatim}
> OR(productmonomials12);
\end{verbatim}\vspace*{-5mm}
\[
{\color{blue} j5 j4 j1^2 + j4 j1 j2 - j3 j1^2 + j1^2 j6}
\]
When two generators commute, like $j1$ and $j2$, we need to use \textsc{SortProducts} with the desired ordering list: $[j2,j1]$
\begin{verbatim}
> orderedproductmonomials12:=SortProducts(SortProducts(OR(productmonomials12),
[j2, j1], usecommutator), [j6, j1], usecommutator);
\end{verbatim}\vspace{-5mm}
\[
{\color{blue} orderedproductmonomials12 := j5 j4 j1^2 + j4 j2 j1 - j3 j1^2 + j6 j1^2}
\]

\newpage

\begin{verbatim}
> orderedproductmonomials21:= OR(productmonomials21);
\end{verbatim}\vspace{-5mm}
\[
{\color{blue} {\it orderedproductmonomials21 :=} j5 j4 j1^2}
\]
Commutator of monomials using the \textsc{Physics Library} command ``Commutator'':
\begin{verbatim}
> commutatormonomials12 := Commutator(monomial1, monomial2);
\end{verbatim}\vspace{-5mm}
\[
{\color{blue}{\it commutatormonomials12:=} j4 j2 j1 + (-j3+j6) j1^2}
\]
Reordering of terms in commutator using function OR and \textsc{SortProducts}
\begin{verbatim}
> commutatormonomials12 := SortProducts(SortProducts(OR(commutatormonomials12),
[j2, j1], usecommutator), [j6, j1], usecommutator);
\end{verbatim}\vspace{-5mm}
\[
{\color{blue}{\it commutatormonomials12:=} j4 j2 j1 - j3 j1^2 + j6 j1^2}
\]
Defining artifacts
\begin{verbatim}
> A[1] := j8*j5-j7*j6:
 A[2] := j8*j3-j7*j4:
 A[3] := j7*j2+j5*j0+j5:
 A[4] := j8*j1+j4*j0+j4:
 A[5] := j7*j1+j3*j0+j3:
 A[6] := j8*j2+j6*j0+j6:
 A[7] := j6*j3-j5*j4+j3:
 A[8] := j6*j1-j4*j2:
 A[9] := j5*j1-j3*j2:
\end{verbatim}
Commutator of artifacts $A[1]$, $A[2]$:
\begin{verbatim}
> CommutatorA1A2 := SortProducts(SortProducts(OR(Commutator(A[1], A[2])),
 [j8, j7], usecommutator), [j7, j4], usecommutator)
\end{verbatim}\vspace{-5mm}
\[
{\color{blue}{\it CommutatorA1A2 :=} -j8^2 j5 + j8 j7 j6 - j8 j7 j3 + j7 j4 j7 - j8 j7}
\]
\begin{verbatim}
 > AnsatzcommutatorA1A2:= '-j8*A[1]-j7*A[2]';
\end{verbatim}\vspace{-5mm}
\[
{\color{blue}{\it AnsatzcommutatorA1A2:=} -j8 A_1 - j7 A_2}
\]
Testing commutator of artifacts $A[1]$, $A[2]$:
\begin{verbatim}
 > OR(Simplify(CommutatorA1A2 - AnsatzcommutatorA1A2));
\end{verbatim}\vspace{-5mm}
\[
{\color{blue}0}
\]

{\bf Note added in proof.} We know that $A_2$ rational system, see \eqref{hA2diff} for the
Hamiltonian $h_3 \equiv h_{{\cal A}_2}(x,y)$ and \eqref{kA2diff} for the cubic
integral $k_3 \equiv k_{{\cal A}_2}(x,y)$ in the algebraic form at
$\mu=\tau=0$, is integrable, $[h_3, k_3]=0$; however, this rational system
admits separation of variables in polar coordinates in the space of
relative motion and, hence, the existence of the additional quadratic
integral of motion $x_3$, see~\cite{TTW:2009}: $[h_3, x_3]=0$ with property $I \equiv
[x_3, k_3] \neq 0$; hence, $[h_3, I]=0$. Thus, the $A_2$ rational
integrable system is superintegrable. Recently, it was shown in \cite{LVT2023} that double commutators
$[x_3, I]=P_3(h_3, k_3, x_3, I)$ and $[k_3, I]=Q_3(h_3, k_3, x_3, I)$ are
cubic polynomials in
$(h_3, k_3, x_3, I)$. Hence, we arrive at 4-generated,
infinite-dimensional cubic polynomial algebra of integrals.

\subsection*{Acknowledgments}

A.V.T.\ is thankful to Willard~Miller Jr.(1937--2023) and Peter~Olver (University of Minnesota, USA) for helpful discussions in different stages of the project and the general encouragement to proceed and to complete this work. Due to enormous computational complexity, this research was running for many years, it was supported in part by the PAPIIT grants IN109512 and IN108815 (Mexico) at the initial stage of the study and by the PAPIIT grant IN113022 (Mexico) at its final stage.
M.A.G.A.\ thanks the CONACyT grant for master degree studies (Mexico) in 2016--2018, when the key calculations of the commutator~(\ref{commutator}) were partially carried out.

All symbolic calculations were carried out by using the MAPLE-18 on a regular DELL desktop computer with CPU processor Intel(R) Core (TM) i7-3770 $@$ 3.40GHz with 6Gb RAM, although some pieces of calculations were done on a regular PC laptop.

A.V.T.\ thanks PASPA-UNAM grant (Mexico) for its support during his sabbatical stay in 2021--2022 at the University of Miami, where this work was mostly completed. We thank all anonymous referees for careful reading of the text, many inspiring comments, remarks and proposals, which improved significantly the presentation.

This work is dedicated to the 70th birthday of Peter Olver to whom we always had admiration as an exemplary mathematician and scientist.

\pdfbookmark[1]{References}{ref}
\LastPageEnding

\end{document}